\documentclass[12pt, a4paper]{article}

\usepackage[margin=0.9in]{geometry}

\usepackage[authoryear]{natbib}

\usepackage[ruled,vlined,linesnumbered]{algorithm2e}

\SetAlFnt{\small}
\SetAlCapFnt{\small}
\SetAlCapNameFnt{\small}
\SetAlCapHSkip{0pt}

\usepackage{booktabs}
\usepackage{amsmath,amsfonts,amsthm,amssymb}
\usepackage[retainorgcmds]{IEEEtrantools}
\usepackage{fourier}
\usepackage{subfigure}
\usepackage{float}
\usepackage{url}
\usepackage{tikz}
\usepackage{changepage}
\providecommand{\DontPrintSemicolon}{\dontprintsemicolon}

\theoremstyle{plain}
\newtheorem{theorem}{Theorem}[section]

\newtheorem{lemma}[theorem]{Lemma}

\newtheorem{myclaim}[theorem]{Claim}

\newtheorem{definition}[theorem]{Definition}
\newtheorem{corollary}[theorem]{Corollary}

\newtheorem{example}{Example}

\DeclareMathOperator*{\argmin}{arg\,min}
\DeclareMathOperator*{\argmax}{arg\,max}

\newcommand{\mms}{\boldsymbol{\mu}}

\newcommand{\mysetminusD}{\hbox{\tikz{\draw[line width=0.6pt,line cap=round] (3pt,0) -- (0,6pt);}}}
\newcommand{\mysetminusT}{\mysetminusD}
\newcommand{\mysetminusS}{\hbox{\tikz{\draw[line width=0.45pt,line cap=round] (2pt,0) -- (0,4pt);}}}
\newcommand{\mysetminusSS}{\hbox{\tikz{\draw[line width=0.4pt,line cap=round] (1.5pt,0) -- (0,3pt);}}}
\newcommand{\mysetminus}{\mathbin{\mathchoice{\mysetminusD}{\mysetminusT}{\mysetminusS}{\mysetminusSS}}}
\allowdisplaybreaks

\title{Approximation Algorithms for Computing Maximin Share Allocations\thanks{A preliminary conference version of this work has appeared in ICALP 2015 \citep{AMNS15}.}}
\author{
	Georgios Amanatidis\thanks{Department of Informatics, Athens University of Economics and Business, Athens, Greece.}\hspace{-10pt} 
	\and
	Evangelos Markakis\footnotemark[2]\hspace{-10pt}
	\and
	Afshin Nikzad\thanks{Department of Management Science and Engineering, Stanford University, CA, USA.}\hspace{-10pt}
	\and
	Amin Saberi\footnotemark[3]
}

\begin{document}

\maketitle

\begin{abstract}
We study the problem of computing maximin share allocations, a recently introduced fairness notion. Given a set of $n$ agents and a set of goods, the maximin share of an agent is the best she can guarantee to herself, if she is allowed to partition the goods in any way she prefers, into $n$ bundles, and then receive her least desirable bundle. 
The objective then is to find a partition, where each agent is guaranteed her maximin share. 
Such allocations do not always exist, hence we resort to approximation algorithms. 
Our main result is a $2/3$-approximation, that runs in polynomial time for any number of agents and goods. This improves upon the algorithm of~\citet{PW14}, which is also a $2/3$-approximation but runs in polynomial time only for a constant number of agents. 
To achieve this, we redesign certain parts of the algorithm in~\citet{PW14}, exploiting the construction of carefully selected matchings in a bipartite graph representation of the problem. 
Furthermore, motivated by the apparent difficulty in establishing lower bounds, we undertake a probabilistic analysis. We prove that 
in randomly generated instances, maximin share allocations exist with high probability. This can be seen as a justification of previously reported experimental evidence.
Finally, we provide further positive results for two special cases arising from previous works. The first is the intriguing case of $3$ agents, where we provide an improved $7/8$-approximation. The second case is when all item values belong to $\{0, 1, 2\}$, where we obtain an exact algorithm.
\end{abstract}

\section{Introduction}
\label{sec:intro}

We study a recently proposed fair division problem in the context of allocating indivisible goods. Fair division has attracted the attention 
of various scientific disciplines, including among others, mathematics, economics, and political science.
Ever since the first attempt for a formal treatment by Steinhaus, Banach, and
Knaster~\citep{Steinhaus48}, many interesting and challenging
questions have emerged. Over the past decades, a vast literature has
developed, see e.g., \citep{BT96,RW98},
and several notions of fairness have been suggested. 
The area gradually gained popularity in computer science as well, as most of the questions are inherently algorithmic, see~\citep{EP84,EP06-focs,WS07}, among others, for earlier works and the surveys by~\citet{Procaccia14-survey} and by~\citet{BCM16-survey} on more recent results.

The objective in fair division problems is to allocate a set of resources to a set of $n$ agents in a way that leaves every agent satisfied.
In the continuous case, the available resources are typically represented by the interval [0, 1], whereas in the discrete case, we have a set of distinct, indivisible goods. 
The preferences of each agent are represented by a valuation function, which is usually an additive function (additive on the set of goods in the discrete case, or a probability distribution on $[0, 1]$ in the continuous case).
Given such a setup, many solution concepts have been proposed as to what constitutes a fair solution. Some of the standard ones include {\it proportionality, envy-freeness, equitability} and several variants of them. 
The most related concept to our work is proportionality, 
where an allocation is called proportional, if each agent receives a bundle of goods that is worth at least $1/n$ of the total value  according to her valuation function.

Interestingly, all the above mentioned solutions and several others can be attained in the continuous case. Apart from mere existence, in some cases we can also have efficient algorithms, see e.g.,~\citep{EP84} for proportionality and \citep{AM16} for some recent progress on envy-freeness. In the presence of indivisible goods however, the picture is quite different. We cannot guarantee existence and it is even NP-hard to decide whether a given instance admits fair allocations. In fact, in most cases it is hard to produce decent approximation guarantees. 

Motivated by the question of what can we guarantee in the discrete case, we focus on a concept recently introduced by~\citet{Budish11}, that can be seen as a relaxation of proportionality. The rationale is as follows: suppose that an agent, say agent
$i$, is asked to partition the goods into $n$ bundles and then the rest of the agents choose a bundle before $i$. In the worst case, agent $i$ 
will be left with her least valuable bundle. Hence, a risk-averse agent would choose a partition that maximizes the minimum value of a 
bundle in the partition. This value is called the maximin share of agent $i$. The objective then is to find an allocation where every 
agent receives at least her maximin share. Even for this notion, existence is not guaranteed under indivisible goods \citep{PW14,KPW16}, despite the encouraging experimental evidence \citep{BL16,PW14}. However, it is possible 
to have constant factor approximations, as has been recently shown~\citep{PW14} (see also our related work section).   \bigskip

\noindent {\bf Contribution:} Our main result, in Section~\ref{sec:2/3}, is a $(2/3-\varepsilon)$-approximation algorithm, for any constant  $\varepsilon >0$, that runs in polynomial time for any number of agents and any number of goods. That is, the algorithm produces an allocation where every agent receives a bundle worth at least $2/3 - \varepsilon$ of her maximin share. Our result improves upon the $2/3$-approximation of~\citet{PW14}, which runs in polynomial time only for a constant number of agents. To achieve this, we redesign certain parts of their algorithm, arguing about the existence of appropriate, carefully constructed matchings in a bipartite graph representation of the problem. 
Before that, in Section~\ref{sec:1/2}, we  provide a much simpler and faster $1/2$-approximation algorithm.  Despite the worse factor, this algorithm still has its own merit due to its simplicity.

Moreover, we study two special cases, motivated by previous works. The first one is the case of $n=3$ agents. This  is an interesting turning point on the approximability of the problem; for $n=2$, there always exist maximin share allocations, but adding a third agent makes the problem significantly more complex, and the best known ratio was $3/4$~\citep{PW14}. We provide an algorithm with an approximation guarantee of $7/8$, by examining more deeply the set of allowed matchings that we can use to satisfy the agents. 
The second case is the setting where all item values belong to $\{0, 1, 2\}$. This is an extension of the $\{0, 1\}$ setting studied by~\citet{BL16} and we show that there always exists a maximin share allocation, for any number of agents.

Finally, motivated by the apparent difficulty in finding impossibility results on the approximability of the problem, we undertake a probabilistic analysis in Section~\ref{sec:random}. Our analysis shows that in randomly generated instances, maximin share allocations exist with high probability. This may be seen as a justification of the reported experimental evidence~\citep{BL16,PW14}, which show that maximin share allocations exist in most cases.\bigskip

\noindent {\bf Related Work:} 
For an overview of the classic fairness notions and related results, we refer the reader to the books of~\citet{BT96}, and~\citet{RW98}. The notion we 
study here was introduced by~\citet{Budish11} for ordinal utilities (i.e., agents have rankings over alternatives), building on concepts by~\citet{Moulin90}. Later on, ~\citet{BL16} defined the notion for cardinal utilities, in the form that we study it here, and  provided many important insights as well as experimental evidence. The first constant factor approximation algorithm was given by~\citet{PW14}, achieving a $2/3$-approximation but in time exponential in the number of agents. 

On the negative side, constructions of instances where no maximin share allocation exists, even for $n=3$, have been provided both by~\citet{PW14}, and by~\citet{KPW16}. These elaborate constructions, along with the extensive experimentation of~\citet{BL16}, reveal 
that it has been challenging to produce better lower bounds, i.e., instances where no $\alpha$-approximation of a maximin share allocation exists, even for 
$\alpha$ very close to $1$. 
Driven by these observations, a probabilistic analysis, similar in spirit but more general than ours, is carried out by~\citet{KPW16}.
In our analysis in Section \ref{sec:random}, all values are uniformly drawn from $[0, 1]$; \citet{KPW16} show a 
similar result with ours but for a a wide range of distributions over $[0, 1]$, establishing that maximin share allocations exist with high probability under all such distributions. 
However, their analysis, general as it may be, needs very large values of $n$ to 
guarantee relatively high probability, hence it does not fully justify the experimental results discussed above. 

Recently, some variants of the problem have also been considered. \citet{BM17} gave a constant factor approximation of $1/10$ for the case where the agents have submodular valuation functions. It remains an interesting open problem to determine whether better factors are achievable for submodular, or other non-additive functions. Along a different direction, \citet{CKMPSW16} introduced the notion of {\it pairwise maximin share guarantee} and provided approximation algorithms. Although conceptually this is not too far apart from maximin shares, the two notions are incomparable.   

Another aspect that has been studied is the design of truthful mechanisms providing approximate maximin share fairness guarantees. Note that our work here does not deal with incentive issues. Looking at this as a mechanism design problem without money,~\citet{ABM16} provide both positive and negative results exhibiting a clear separation between what can be achieved with and without the truthfulness constraint. Even further, \citet{ABCM17} completely characterized truthful mechanisms for two agents, which in turn implied tight bounds on the approximability of maximin share fairness by truthful mechanisms.    

Finally, a seemingly related problem is that of max-min fairness (also known as the Santa Claus problem)~\citep{AS07,BS06,BD05}. In this problem we want to find an allocation where the value of the least happy person is maximized. With identical agents, this coincides with our problem, but beyond this special case the two problems exhibit very different behavior.

\section{Definitions and Notation}
\label{sec:defs}

For any $k\in\mathbb N$, we denote by $[k]$ the set $\{1,\ldots ,k\}$.
Let $N = [n]$ be a set of $n$ agents and $M = [m]$ be a set of indivisible items.
Following the usual setup in the fair division literature, we assume each agent has an additive valuation function $v_i(\cdot)$, so that for every $S\subseteq M$, 
$v_i(S) = \sum_{j\in S} v_i(\{j\})$. For $j\in M$, we will use $v_{ij}$ instead of $v_i(\{j\})$.

Given any subset $S\subseteq M$, an allocation of $S$ to the $n$ agents is a partition $T = (T_1,...,T_n)$, where $T_i\cap T_j = \emptyset$ and $\bigcup T_i = S$.
Let $\Pi_n(S)$ be the set of all partitions of a set $S$ into $n$ bundles.

\begin{definition}
Given a set of $n$ agents, and any set $S\subseteq M$, the $n$-maximin share of an agent $i$ with respect to $S$, is:
\[ \mms_i(n, S) = \displaystyle\max_{T\in\Pi_n(S)} \min_{T_j\in T} v_i(T_j) \, .\]
\end{definition}

Note that $\mms_i(n, S)$ depends on the valuation function $v_i(\cdot)$ but is independent of any other function $v_j(\cdot)$ for $j\neq i$. When $S=M$, we refer to $\mms_i(n, M)$ as the maximin share of agent $i$. The solution concept we study asks for a partition that gives each agent her maximin share. 
\begin{definition}
Given a set of agents $N$, and a set of goods $M$, a partition $T = (T_1,...,T_n)\allowbreak \in \Pi_n(M)$ is called a maximin share (MMS) allocation if $v_i(T_i)\geq \mms_i(n, M)\,$, for every agent $i\in N$.
\end{definition} 

Before we continue, a few words are in order regarding the appeal of this new concept. First of all, it is very easy to see that having a maximin share guarantee to every agent forms a relaxation of proportionality, see Claim \ref{cl:upper_bound_1}. Given the known impossibility results for proportional allocations under indivisible items, it is worth investigating whether such relaxations are easier to attain. Second, the maximin share guarantee has an intuitive interpretation; for an agent $i$, it is the value that could be achieved if we run the generalization of the cut-and-choose protocol for multiple agents, with $i$ being the cutter. In other words, it is the value that agent $i$ can guarantee to himself, if he were given the advantage to control the partition of the items into bundles, but not the allocation of the bundles to the agents.

\begin{example}
\label{ex:mms}
Consider an instance with three agents and five items: 
\begin{center}
\begin{tabular}{@{} *6c @{}}    
 & $a$ & $b$ & $c$ & $d$ & $e$ \\\toprule
\ Agent 1 & $1/2$ & $1/2$ & $1/3$ & $1/3$ & $1/3$ \ \\ 
\ Agent 2 & $1/2$ & $1/4$ & $1/4$ & $1/4$ & $0$\ \\ 
\ Agent 3 & $1/2$ & $1/2$ & $1$ & $1/2$ & $1/2$\ \\\bottomrule
\end{tabular}\vspace{3pt}
\end{center}
If $M = \{a, b, c, d, e\}$ is the set of items, one can see that $\mms_1(3, M) = 1/2$, $\mms_2(3, M) = 1/4$, $\mms_3(3, M) = 1$. E.g., for agent $1$, no matter how she partitions the items into three bundles, the worst bundle will be worth at most $1/2$ for her, and she achieves this with the  partition $(\{a\}, \{b, c\}, \{d, e\})$. 
Similarly, agent $3$ can guarantee a value of $1$ (which is best possible as it is equal to $v_3(M)/n$) by the partition $(\{a, b\}, \{c\}, \{d, e\})$. 

Note that this instance admits a maximin share allocation, e.g., $(\{a\}, \{b, c\}, \{d, e\})$, and in fact this is not unique.
Note also that if we remove some agent, say agent 2, the maximin values for the other two agents increase. E.g., $\mms_1(2, M) = 1$, achieved by the partition $(\{a, b\}, \{c, d, e\})$. Similarly, $\mms_3(2, M) = 3/2$.  
\qed
\end{example}

As shown in~\citep{PW14}, maximin share allocations do not always exist. Hence, our focus is on approximation algorithms, i.e., on algorithms that
produce a partition where each agent $i$ receives a bundle worth (according to $v_i$) at least $\rho \cdot \mms_i(n, M)$, for some $\rho\leq 1$.

\section{Warmup: Some Useful Properties and a Polynomial Time $1/2$-approximation}
\label{sec:1/2}

We find it instructive to provide first a simpler and faster algorithm that achieves a worse approximation of $1/2$. In the course of obtaining this algorithm, we also identify some important properties and insights that we will use in the next sections.
  
We start with an upper bound on our solution for each agent. The maximin share guarantee is a relaxation of proportionality, so we trivially have:

\begin{myclaim}
\label{cl:upper_bound_1}
For every $i\in N$ and every $S\subseteq M$, $\mms_i(n,S) \leq \displaystyle\frac{v_i(S)}{n} = \frac{\sum_{j\in S}v_{ij}}{n}$.
\end{myclaim}
\begin{proof}
This follows by the definition of  maximin share. If there existed a partition where the minimum value for agent $i$  exceeded the above bound, then the total value for agent $i$ would be more than $\sum_{j\in S} v_{ij}$. 
\end{proof}

Based on this, we now show how to get an additive approximation. Algorithm 1 below achieves an additive approximation of $v_{max}$, where $v_{max} = \max_{i, j} v_{ij}$. This simple algorithm, which we will refer to as the {\em Greedy Round-Robin Algorithm}, has also been discussed by~\citet{BL16}, where it was shown that when all item values are in $\{0, 1\}$, it produces an exact maximin share allocation. At the same time, we note that the algorithm also achieves envy-freeness up to one item, another solution concept defined by~\citet{Budish11}, and further discussed in~\citet{CKMPSW16}. Finally, some variations of this algorithm have also been used in other allocation problems, see e.g.,~\citet{BK05}, or the protocol in
\citet{BL11}.  
We discuss further the properties of Greedy Round-Robin in Section~\ref{sec:random}. 

In the statement of the algorithm below, the set $V_N$ is the set of valuation functions $V_N = \{v_i : i\in N\}$, which can be encoded as a valuation matrix since the functions are additive. \vspace{5pt}

\begin{algorithm}[H]
\DontPrintSemicolon 
Set $S_i=\emptyset$ for each $i\in N$. \;
Fix an ordering of the agents arbitrarily. \;
\While{$\exists$ unallocated items}{
$S_i = S_i\cup \{j\}$, where $i$ is the next agent to be examined in the current round (proceeding in a round-robin fashion) and $j$ is $i$'s most desired item among the currently unallocated items. \;
}
\textbf{return} $(S_1,...,S_n)$ \;
\caption{Greedy Round-Robin$(N, M, V_N)$} \label{fig:alg-1}
\end{algorithm} 

\begin{theorem}
\label{thm:additive}
If $(S_1,...,S_n)$ is the output of Algorithm \ref{fig:alg-1}, then for every $i\in N$,
\begin{equation}
\label{eq:additive}
 v_i(S_i) \geq \frac{\sum_{j\in M} v_{ij}}{n} - v_{max} \geq \mms_i(n, M) - v_{max}\nonumber \,.
 \end{equation}
\end{theorem}

\begin{proof}
Let $(S_1,...,S_n)$ be the output of Algorithm 1. We first prove the following claim about the envy of each agent towards the rest of the agents:

\begin{myclaim}
\label{cl:envy}
For every $i, j\in N$, $v_i(S_i) \geq v_i(S_j) - v_{max}$.
\end{myclaim}
\begin{proof} \renewcommand{\qedsymbol}{{\footnotesize $\boxdot$}}
Fix an agent $i$, and let $j\neq i$. We will upper bound the difference 
 $v_i(S_j) - v_i(S_i)$. If $j$ comes after $i$ in the order chosen by the algorithm, then the statement of the claim trivially holds, since $i$ always picks an item at least as desirable as the one $j$ picks.
Suppose that $j$ precedes $i$ in the ordering. The algorithm proceeds in $\ell = \lceil m/n \rceil$ rounds. In each round $k$, let $r_k$ and $r_k'$ be the items allocated to $j$ and $i$ respectively. Then
\[ v_i(S_j) - v_i(S_i) = (v_{i, r_1} - v_{i, r_1'}) + (v_{i, r_2} - v_{i, r_2'}) + \cdots + (v_{i, r_{\ell}} - v_{i, r_{\ell}'})\,.\]  
Note that there may be no item $r_{\ell}'$ in the last round if the algorithm runs out of goods but this does not affect the analysis (simply set $v_{i, r_{\ell}'} = 0$).

Since agent $i$ picks her most desirable item when it is her turn to choose, this means that for two consecutive rounds $k$ and $k+1$ it holds that
$v_{i, r_k'} \geq v_{i, r_{k+1}}$. This directly implies that
\[ v_i(S_j) - v_i(S_i) \leq v_{i, r_1} -  v_{i, r_{\ell}'} \leq v_{i, r_1} \leq v_{max}\,.\qedhere\]
\end{proof}

If we now sum up the statement of Claim \ref{cl:envy} for each $j$, we get: $n v_i(S_i) \geq \sum_j v_i(S_j) - n v_{max}$, which implies
\[ v_i(S_i) \geq \frac{\sum_j v_i(S_j)}{n} - v_{max} = \frac{\sum_{j\in M} v_{ij}}{n} - v_{max} \geq \mms_i(n, M) - v_{max}\,, \]
where the last inequality holds by Claim \ref{cl:upper_bound_1}. 
\end{proof}

The next important ingredient is the following monotonicity property, which says that we can allocate a single good to an agent without decreasing the maximin share of other agents. Note that this lemma also follows from Lemma 1 of \citet{BL16}, yet, for completeness, we prove it here as well.

\begin{lemma}[Monotonicity property]
\label{lem:monotonicity}
For any agent $i$ and any good $j$, it holds that 
\[\mms_i(n-1, M\mysetminus \{j\}) \geq \mms_i(n, M)\,.\]
\end{lemma} 
\begin{proof}
Let us look at agent $i$, and consider a partition of $M$ that attains her maximin share.
Let $(S_1,..., S_n)$ be this partition. Without loss of generality, suppose $j\in S_1$. Consider the remaining partition $(S_2,...,S_n)$ enhanced in an arbitrary way by the items of $S_1\mysetminus\{j\}$.
This is a $(n-1)$-partition of $M\mysetminus \{j\}$ where the value of agent $i$ for any bundle is at least $\mms_i(n, M)$. Thus, 
we have $\mms_i(n-1, M\mysetminus \{j\}) \geq \mms_i(n, M)$.
\end{proof}

We are now ready for the $1/2$-approximation, obtained by Algorithm~\ref{fig:alg-1/2} below, which is based on
using Greedy Round-Robin, but only after we allocate first the most valuable goods. This is done so that the value of $v_{max}$ drops to an extent that Greedy Round-Robin can achieve a multiplicative approximation. \vspace{5pt}

\begin{algorithm}[H]
\DontPrintSemicolon 
Set $S =M$ \;
\For{$i = 1$ to $|N|$} {
Let $\alpha_i = \frac{\sum_{j\in S} v_{ij}}{|N|}$  \;   
}
\While{$\exists i, j$ s.t. $v_{ij}\geq \alpha_i/2$}{ \label{line:1st-while}
Allocate $j$ to $i$. \;
$S = S\mysetminus \{j\}$ \;
$N = N\mysetminus \{i\}$ \;
Recompute the $\alpha_i$s.  \label{line:last-while} \;
} 
Run Greedy Round-Robin on the remaining instance. \; 
\caption{$\textsc{apx-mms}_{1/2}(N, M, V_N)$}\label{fig:alg-1/2}
\end{algorithm}

\begin{theorem}
\label{thm:1/2}
Let $N$ be a set of $n$ agents, and let $M$ be a set of goods. Algorithm~\ref{fig:alg-1/2} produces an allocation $(S_1,...,S_n)$  such that 
\[ v_i(S_i) \geq \frac{1}{2} \mms_i(n, M)\,,\ \forall i\in N\,.\]
\end{theorem}

\begin{proof}
We will distinguish two cases. Consider an agent $i$ who was allocated a single item during the first phase of the algorithm (lines \ref{line:1st-while} - \ref{line:last-while}). Suppose that at the time when $i$ was given her item, there were $n_1$ active agents, $n_1\leq n$, and that $S$ was the set of currently unallocated items. By the design of the algorithm, this means that the value of what $i$ received is at least 
 \[ \frac{\sum_{j\in S} v_{ij}}{2n_1} \geq \frac{1}{2}\mms_i(n_1, S)\]
where the inequality follows by Claim~\ref{cl:upper_bound_1}. 
But now if we apply the monotonicity property (Lemma~\ref{lem:monotonicity}) $n-n_1$ times, we get that $\mms_i(n_1, S) \geq \mms_i(n, M)$, and we are done.

Consider now an agent $i$, who gets  a bundle of goods according to Greedy Round-Robin, in the second phase of the algorithm. Let $n_2$ be the number of active agents at that point, and $S$ be the set of goods that are unallocated before Greedy Round-Robin is executed. We know that $v_{max}$ at that point is less than half the current value of $\alpha_i$ for agent $i$. Hence by the additive guarantee of Greedy Round-Robin, we have that the bundle received by agent $i$ has value at least
\[ \frac{\sum_{j\in S} v_{ij}}{n_2} - v_{max} > \frac{\sum_{j\in S} v_{ij}}{n_2} - \frac{\alpha_i}{2} = \frac{\sum_{j\in S} v_{ij}}{2n_2} \geq \frac{1}{2}\mms_i(n_2, S)\,.\]
Again, after applying the monotonicity property repeatedly, we get that $\mms_i(n_2, S) \geq \mms_i(n, M)$, which completes the proof.
\end{proof}

\section{A Polynomial Time $\big(\frac{2}{3}-\varepsilon\big)$-approximation}
\label{sec:2/3}

The main result of this section is  Theorem \ref{thm:2/3}, establishing a polynomial time algorithm for achieving a $2/3$-approximation to the maximin share of each agent.

\begin{theorem}\label{thm:2/3}
Let $N$ be a set of $n$ agents, and let $M$ be a set of goods. 
For any constant $\varepsilon>0$, Algorithm \ref{fig:alg-pw} produces in polynomial time an allocation $(S_1,...,S_n)$, 
such that
\[ v_i(S_i) \geq \left(\frac{2}{3}-\varepsilon\right) \mms_i(n, M)\,,\ \forall i\in N\,.\]
\end{theorem}

Our result is based on the algorithm by~\citet{PW14}, which also guarantees to each agent a $2/3$-approximation. However, their algorithm runs in polynomial time only for a constant number of agents. Here, we identify the source of exponentiality and take a different approach regarding certain parts of the algorithm.
For the
sake of completeness, we first present  the necessary 
related results
of~\citet{PW14}, before we discuss the steps that are needed to obtain our result.

First of all, we note that even the computation of the maximin share values is already a hard problem. For a single agent $i$, the problem of deciding whether $\mms_i(n, M)\ge k$ for a given $k$ is NP-complete. 
However, a PTAS follows by the work of~\citet{Woeginger97}. In the original paper, which is in the context of job scheduling, Woeginger gave a PTAS for maximizing the minimum completion time on identical machines. But this scheduling problem is identical to computing a maximin partition with respect to a given agent $i$.
Indeed, from agent $i$'s
perspective, it is enough to think of the machines as identical agents (the only input that we need for computing $\mms_i(n, M)$ is the valuation function of $i$). Hence:

\begin{theorem}[\normalfont{Follows by \citep{Woeginger97}}]
\label{thm:woeginger}
Suppose we have a set $M$ of goods to be divided among $n$ agents. Then, for each agent $i$, there exists a PTAS for approximating $\mms_i(n, M)$.
\end{theorem}


A central quantity in the algorithm of~\citet{PW14} is the $n$-{\em density balance parameter}, denoted by $\rho_n$ and defined below. 
Before stating the definition, we give for clarity the high level idea, which can be seen as an attempt to generalize the monotonicity property of Lemma \ref{lem:monotonicity}. Assume that in the course of an algorithm, 
we have used a subset of the items to ``satisfy'' some of the agents, and that those items do not have \lq\lq{}too much\rq\rq{} value 
for the rest of the agents. If $k$ is the number of remaining agents, and $S$ is the remaining set of goods, then we should expect to be able to ``satisfy'' these $k$ agents using the items in $S$. A good approximation in this reduced instance however, would only be an approximation with respect to $\mms_i(k, S)$. Hence, in order to hope for an approximation algorithm for the original instance, we would need to examine how $\mms_i(k, S)$ relates to $\mms_i(n, M)$. Essentially, the parameter $\rho_n$ is the best guarantee one can hope to achieve for the remaining agents, based only on the fact that the complement of the set left to be shared is of relatively small value. Formally:

\begin{definition}[\emph{\citep{PW14}}]
For any number $n$ of agents, let 
\[ \rho_n = \max \left\{ \lambda\ \Bigg | 
	\begin{array}{l}
		\forall  M, \forall \textnormal{ additive } v_i\in  ({\mathbb R}^+)^{2^{M}}, \forall S\subseteq M, 
		\forall k, \ell \textnormal{ s.t. } k+\ell=n, \\
		v_i(M \mysetminus S) \le \ell\lambda \mms_i(n, M) \Rightarrow \mms_i(k, S) \ge \lambda \mms_i(n, M)
	\end{array}
\right\}\,. \]
\end{definition}

After a quite technical analysis, Procaccia and Wang calculate the exact value of $\rho_n$ in the following lemma.
\begin{lemma}[\normalfont{Lemma 3.2 of~\citep{PW14}}]
For any $n\ge 2$,
\[ \rho_n = \frac{2\lfloor n \rfloor_{odd}}{3\lfloor n \rfloor_{odd}-1} > \frac{2}{3}\, , \]
where $\lfloor n \rfloor_{odd}$ denotes the largest odd integer less than or equal to $n$.
\end{lemma}

We are now ready to state our algorithm, referred to as $\textsc{apx-mms}$ (Algorithm~\ref{fig:alg-pw} below). 
We elaborate on the crucial differences between Algorithm~\ref{fig:alg-pw} and the result of~\citet{PW14} 
after the algorithm description (namely after Lemma~\ref{lem:x+}).
At first, the algorithm computes each agent's $(1-\varepsilon')$-approximate
maximin value using Woeginger's PTAS, where $\varepsilon'=\frac{3\varepsilon}{4}$. 
Let $\boldsymbol{\xi}=(\xi_1, \ldots, \xi_n)$ be the vector of 
these values. Hence, $\forall i,\ \mms_i(n, M)\ge \xi_i\ge (1-\varepsilon') \mms_i(n, M)$.  Then, $\textsc{apx-mms}$ makes a call to the recursive algorithm 
$\textsc{rec-mms}$ (Algorithm \ref{fig:alg-pw_rec}) to compute a $\left(\frac{2}{3}-\varepsilon\right)$-approximate partition. 
$\textsc{rec-mms}$ takes the arguments $\varepsilon', n=|N|$, $\boldsymbol{\xi}$, $S$ (the set of items that have not been allocated 
yet), $K$ (the set of agents  that have not received a share of items yet), and the valuation functions $V_K=\{v_i | i\in K \}$. 
The guarantee provided by $\textsc{rec-mms}$ is that as long as the already allocated goods are not worth too much for the currently active agents of $K$, we can satisfy them with the remaining goods. More formally, under the assumption that 
\begin{equation}\label{recursion-balance}
	\forall i\in K,\ \ v_i (M\mysetminus S) \le (n-|K|)\rho_n \mms_i(n, M)\,,
\end{equation}
which we will show that it holds before each call, $\textsc{rec-mms}(\varepsilon', n, \boldsymbol{\xi}, S, K, V_K)$ computes a $|K|$-partition of $S$, so that each agent receives 
items of value at least $(1-\varepsilon')\rho_n \xi_i$. 

The initial call of the recursion is, of course, $\textsc{rec-mms}(\varepsilon', n, \boldsymbol{\xi}, M, N, V_N)$.
Before moving on to the next recursive call, $\textsc{rec-mms}$ appropriately allocates some of the items to some of the agents,
so that they receive value at least $(1-\varepsilon')\rho_n \xi_i$ each. This is achieved by identifying an appropriate matching between some currently unsatisfied agents and certain bundles of items, as described in the algorithm. In particular, the most important step in the algorithm is to first compute the set $X^+$ (line \ref{line:x+}), which is the set of agents that will 
not be matched in the current call. 
The remaining active agents, i.e., $K\mysetminus X^+$, are then guaranteed to get matched in the current round, whereas $X^+$ will be satisfied in the next recursive calls.  In order to ensure this for $X^+$, 
$\textsc{rec-mms}$ guarantees that inequality 
\eqref{recursion-balance} holds for $K = X^+$ and with $S$ being the rest of the items. Note that \eqref{recursion-balance} trivially holds for 
the initial call of $\textsc{rec-mms}$, where $K = N$ and $S = M$. \vspace{5pt}

\begin{algorithm}[H]
	\DontPrintSemicolon 
		$\varepsilon'=\frac{3\varepsilon}{4}$ \;
		\For{$i = 1$ to $|N|$}{ 
			Use Woeginger's PTAS to compute a $(1-\varepsilon')$-approximation $\xi_i$ of $\mms_i(|N|, M)$. Let $\boldsymbol{\xi}=(\xi_1, \ldots, \xi_n)$. \; 
		}
		\Return $\textsc{rec-mms}(\varepsilon', |N|, \boldsymbol{\xi}, M, N, V_N)$ \;
	\caption{$\textsc{apx-mms}(\varepsilon, N, M, V_N)$}\label{fig:alg-pw}
\end{algorithm}\vspace{5pt}

For simplicity, in the description of $\textsc{rec-mms}$, we assume that $K=\{1, 2, \ldots$, $|K|\}$.
Also, for the bipartite graph defined below in the algorithm, by $\Gamma(X^+)$ we denote the set of neighbors of the vertices in $X^+$. \vspace{5pt}

\begin{algorithm}[H]
	\DontPrintSemicolon 
		\eIf{$|K|=1$}{ 
			Allocate all of $S$ to agent 1.\; }
		{
			Use Woeginger's PTAS to compute a $(1-\varepsilon')$-approximate $|K|$-maximin 
			partition of $S$ with respect to agent $1$ from $K$, say ($S_1, \ldots, S_{|K|}$). \; \label{line:PTAS}
			Create a bipartite graph $G=(X\cup Y, E)$, where $X=Y=K$ and 
			$E=\{(i,j)\,|\,  i\in X,  j\in Y, \allowbreak v_i(S_j)\ge (1-\varepsilon')\rho_n \xi_i\}$. \;\label{line:G}
			Find a set $X^+\subset X$, as described in Lemma \ref{lem:x+}. \;\label{line:x+}
			Given a perfect matching $A$, between $X\mysetminus X^+$ and a subset of 
			$Y\mysetminus \Gamma(X^+)$, allocate $S_j$ to agent $i$ iff $(i, j)\in A$
			(the matching is a byproduct of line \ref{line:x+}). \; \label{line:allocate}
			\eIf{$X^+=\emptyset$}{
				Output the above allocation.\;}
			{
				Output the above allocation, together with $\textsc{rec-mms}(\varepsilon', n, \boldsymbol{\xi},  S^*,  X^+, V_{X^+})$,
				where $S^*$ is the subset of $S$ not allocated in line \ref{line:allocate}.\;
			}
		}
	\caption{$\textsc{rec-mms}(\varepsilon', n, \boldsymbol{\xi}, S, K, V_K)$}\label{fig:alg-pw_rec}
\end{algorithm} \vspace{5pt}

To proceed with the analysis, and since the choice of $X^+$ plays an important role (line \ref{line:x+} of Algorithm \ref{fig:alg-pw_rec}),
we should first clarify what properties of $X^+$ are needed for 
the algorithm to work. The following lemma is the most crucial part in the design of our algorithm.
\begin{lemma}\label{lem:x+}
Assume that for $n$, $M$, $S$, $K$, $V_K$ inequality \eqref{recursion-balance} holds and let $G = (X\cup Y, E)$ be the bipartite graph defined
in line \ref{line:G} of $\textsc{rec-mms}$. Then there exists a subset $X^+$ of $X\mysetminus \{1\}$, such that: 
\begin{itemize}
  \item[(i)] $X^+$ can be found efficiently.
  \item[(ii)] There exists a perfect matching between $X\mysetminus X^+$ and a subset of 
  							 	   $Y\mysetminus \Gamma(X^+)$.
  \item[(iii)] If we allocate subsets to agents according to such a matching (as described in line 
  									\ref{line:allocate}) and $X^+\neq\emptyset$, then inequality \eqref{recursion-balance} 
  									holds for $n$, $M$, $S^*$, $X^+$, $V_{X^+}$ where $S^*\subseteq S$ is the unallocated set of items, i.e.:
 \[	\forall i\in X^+,\ \ v_i (M\mysetminus S^*) \le (n-|X^+|)\rho_n \mms_i(n, M)\,.\]
\end{itemize}
\end{lemma}

Before we prove Lemma \ref{lem:x+}, we elaborate on the main differences between our setup and the approach of~\citet{PW14}:\medskip

\noindent\textit{\textbf{Choice of $X^+$.}} In~\citet{PW14}, $X^+$ is defined as
${\arg \max}_{Z\subseteq K\mysetminus \{1\}}\{|Z|\,|\, |Z|\ge |\Gamma(Z)|\}$. Clearly, when $n$ is constant, so is $|K|$, and 
thus the computation of $X^+$ is trivial. However, it is not clear how to efficiently find such a set in general, when $n$ is not constant. 
We propose a definition of $X^+$, which is efficiently computable and has the desired properties. 
In short, our $X^+$ is any 
appropriately selected counterexample to Hall's Theorem for the graph $G$ constructed in line \ref{line:G}. \medskip

\noindent\textit{\textbf{Choice of $\varepsilon$.}} 
The algorithm works for any $\varepsilon > 0$, but~\citet{PW14} choose an $\varepsilon$ that depends on $n$, and it is such that
$(1-\varepsilon)\rho_n \ge \frac{2}{3}$. This is possible since for any $n,\ \rho_n \ge \frac{2}{3}\big(1+\frac{1}{3n-1}\big)$. 
However, in this case, the running time of Woeginger's PTAS (line \ref{line:PTAS}) is not polynomial in $n$. Here, we consider any fixed $\varepsilon$, independent of $n$, hence the approximation ratio of $\frac{2}{3}-\varepsilon$.\smallskip

The formal definition of $X^+$ is given within the proof of 
Lemma~\ref{lem:x+} that follows. 

\begin{proof}[Proof of Lemma \ref{lem:x+}]
We will show that either $X^+=\emptyset$ (in the case where $G$ has a perfect matching), or some set 
$X^+$ with 
$X^+ \in \big\{Z\subseteq X \,:\, |Z|>|\Gamma(Z)| \wedge \exists ~\allowbreak\text{matching of size } 
|X\mysetminus Z| \text{ in } G\mysetminus \{Z\cup\Gamma(Z)\} \big\}$ has the desired properties. 
Moreover, we propose  a way to find such a set efficiently. We first find a maximum matching $B$ of $G$.
If $|B|=|K|$, then we are done, since for $X^+=\emptyset$, properties \textit{(i)} and \textit{(ii)} of Lemma \ref{lem:x+} hold, while we need not check \textit{(iii)}. 
If $|B|<|K|$, then there must
be a subset of $X$ violating the condition of Hall's Theorem.\footnote{The special case of Hall's Theorem \citep{Hall35} used here, states that given a bipartite graph $G = (X\cup Y, E)$, where $X, Y$ are disjoint independent sets with $|X|=|Y|$, there is a perfect matching in $G$ if and only if $|W|\le |\Gamma(W)|$ for every $W \subseteq X$.} 
Let $X_u, X_m$ be the partition of $X$ in unmatched and matched vertices
respectively, according to $B$, with $X_u\neq \emptyset$, $X_m\neq \emptyset$. Similarly, we define $Y_u, Y_m$. 

We now construct a directed graph $G'=(X\cup Y, E')$, where we direct all edges of $G$ from $X$ to $Y$, and on top of that, we add one copy of each edge of the matching but with direction from $Y$ to $X$. In particular, $\forall i \in X, \forall j \in Y$, if $(i, j)\in E$ then $(i, j)\in E'$, and moreover if $(i, j)\in B$ then $(j, i)\in E'$.
We claim that the following set satisfies the desired properties
\begin{equation}
\label{def_x+}
X^+ := X_u \cup \{v\in X: v \mbox{ is reachable from } X_u \mbox{ in } G'\}\, .\nonumber
\end{equation}
Note that $X^+$ is
easy to compute; after finding the maximum matching in $G$, and constructing $G'$, we can run a depth-first search in each connected component of $G'$, starting from the vertices of $X_u$. See also Figure~\ref{fig:X+}, after the proof of Theorem~\ref{thm:2/3} for an illustration. 

Given the definition of $X^+$, 
we now show property \textit{(ii)}. 
Back to the original graph $G$, we first claim that $|X^+|>|\Gamma(X^+)|$. To prove this, note that 
if $j\in \Gamma(X^+)$ in $G$, then $j\in Y_m$. If not, 
then it is not difficult to see that there is an augmenting path from a vertex in $X_u$ to $j$, which contradicts the maximality of $B$. Indeed, since $j\in \Gamma(X^+)$, let $i$ be a neighbor of $j$ in $X^+$. If $i\in X_u$, then the edge $(i, j)$ would enlarge the matching. Otherwise, $i\in X_m$ and since also $i\in X^+$, there is a path in $G'$ from some vertex of $X_u$ to $i$. But this path by construction of the directed graph $G'$ must consist of an alternation of unmatched and matched edges, hence together with $(i, j)$ we have an augmenting path.

Therefore, $\Gamma(X^+) \subseteq Y_m$, i.e., for any $j\in\Gamma(X^+)$, there is an edge $(i, j)$ in the matching $B$. But then $i$ has to belong to $X^+$ by the construction of $G\rq{}$ (and since $j\in\Gamma(X^+)$).
To sum up: for any $j \in \Gamma(X^+)$, there is exactly one distinct vertex $i$, with $(i, j)\in E$, and $i \in X^+\cap X_m$, i.e., $|X^+\cap X_m|\geq|\Gamma(X^+)|$. In fact, we have equality here, because it is also true that for any $i\in X^+\cap X_m$, there is a distinct vertex $j\in Y_m$ which is trivially reachable from $X^+$. Hence, $|X^+\cap X_m| = |\Gamma(X^+)|$. 
Since $X_u\neq \emptyset$, we have 
$|X^+|=|X_u| + |X^+\cap X_m|\ge 1 + |\Gamma(X^+)|$.
So, $|X^+|>|\Gamma(X^+)|$. 

Also, note that $X^+ \subseteq X\mysetminus \{1\}$, because for any $Z\subseteq X$ that contains
vertex 1 we have $|\Gamma(Z)|=|K|\ge |Z|$. This is due to the fact that for any vertex $j\in Y$, the edge $(1, j)$ is present by the construction, since $v_1(S_j)\ge (1-\varepsilon') \mms_1(k, S) \ge (1-\varepsilon') \rho_n \mms_1(n, M)\ge (1-\varepsilon') \rho_n \xi_1$, for all $1\le j \le |K|$.

We now claim that if we remove $X^+$ and $\Gamma(X^+)$ from $G$, then the restriction of $B$ on the remaining graph, still 
matches all vertices of $X\mysetminus X^+$, establishing property \textit{(ii)}. Indeed, note first that for any $i\in X\mysetminus X^+$, it has to hold that $i\in X_m$, since $X^+$ contains $X_u$. Also, for any edge $(i, j)\in B$ with $i\in X$ and $j\in \Gamma(X^+)$, 
we have $i\in X^+$ by the construction of $X^+$. So, for any $i\in X\mysetminus X^+$, its pair in $B$ belongs to 
$Y\mysetminus \Gamma(X^+)$. Equivalently, $B$ induces a perfect matching between $X\mysetminus X^+$ and a subset of 
$Y\mysetminus \Gamma(X^+)$ (this is the matching $A$ in line \ref{line:allocate} of the algorithm).

What is left to prove is that property \textit{(iii)} also holds for $X^+$. This can be done by the same arguments as in~\citet{PW14}, specifically by the following lemma which can be inferred from their work.

\begin{lemma}[\normalfont{\citep{PW14}, end of Subsection 3.1}]\label{lem:balance}
	Assume that inequality \eqref{recursion-balance} holds for $n$, $M$, $S$, $K$, $V_K$, and let $G$ be the graph defined
	in line \ref{line:G}. For any $Z\subseteq X$, if there exists a perfect matching between 
	$X\mysetminus Z$ and a subset of $Y\mysetminus \Gamma(Z)$, say $Y^*$, and there are no edges between 
	$Z$ and $Y^*$ in $G$, then property \textit{(iii)} holds as well.
\end{lemma}

Clearly, there are no edges between $X^+$ and $Y\mysetminus \Gamma(X^+)$. Hence, Lemma \ref{lem:balance} can be applied to $X^+$, completing the proof. 
\end{proof}

Given Lemma \ref{lem:x+}, we can now prove the main result of this section, the correctness of $\textsc{apx-mms}$. 

\begin{proof}[Proof of Theorem \ref{thm:2/3}]
It is clear that the running time of the algorithm is polynomial. Its correctness is based on the correctness of $\textsc{rec-mms}$. The latter can be proven with strong 
induction on $|K|$, the number of still active agents that $\textsc{rec-mms}$ receives as input, under the assumption that \eqref{recursion-balance} holds before each new call of $\textsc{rec-mms}$ (which we have established by Lemma \ref{lem:x+}). For $|K|=1$, assuming that inequality \eqref{recursion-balance} holds, we have for agent $1$ of $K$:
\begin{IEEEeqnarray*}{rCl}
v_1(S) & =&  v_1(M) - v_1(M\mysetminus S) \ge  n \mms_1(n, M) - (n-1)\rho_n \mms_1(n, M)  \\
&\ge& \mms_1(n, M) \ge \left( \frac{2}{3}-\varepsilon \right) \mms_1(n, M).
\end{IEEEeqnarray*}

For the inductive step, Lemma \ref{lem:x+} and the choice of $X^+$ are crucial. 
Consider an execution of $\textsc{rec-mms}$ during which some agents will receive a subset of items and the rest will form the set $X^+$ to be handled recursively. 
For all the agents in $X^+$ --if any--
we are guaranteed $\big(\frac{2}{3}-\varepsilon\big)$-approximate shares by property \textit{(iii)} of Lemma \ref{lem:x+} and by the inductive hypothesis.
On the other hand, for each agent $i$ that receives a subset $S_j$ of items in line \ref{line:allocate}, we have 
\[v_i(S_j)\ge (1-\varepsilon')\rho_n \xi_i\ge (1-\varepsilon')^2 \rho_n\mms_i(n, M) > (1-2\varepsilon')\frac{2}{3}\mms_i(n, M)=\left( \frac{2}{3}-\varepsilon\right)\mms_i(n, M)\,,\]
where the first inequality holds because $(i, j)\in E(G)$. 
\end{proof}

In Figure~\ref{fig:X+}, we give a simple snapshot to illustrate a recursive call of $\textsc{rec-mms}$. 
In particular, in Subfigure~\ref{subfig:G}, we see a bipartite graph $G$ that could be the current configuration for $\textsc{rec-mms}$, 
along with a maximum matching.
In Subfigure~\ref{subfig:G\rq{}}, we see the construction of $G\rq{}$, as described in Lemma~\ref{lem:x+}, and the set $X^+$.
The bold (black) edges in $G'$ signify that both directions are present. The set $X^+$ consists then of $X_u$ and all other vertices of $X$ reachable from $X_u$.
Finally, Subfigure~\ref{subfig:G\rq{}} also shows the set of agents that are satisfied in the current call along with the corresponding perfect matching, as claimed in Lemma~\ref{lem:x+}.

\begin{figure}[h]
\centering
\subfigure[The graph $G$ defined in line \ref{line:G} of Algorithm \ref{fig:alg-pw_rec} shown with a maximum matching (blue edges). Agent 1 is the top vertex of $X$.]{
\mbox{\includegraphics[scale = 1]{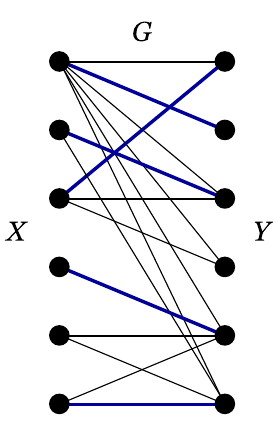}}
\label{subfig:G}
}
\hspace{.05in}
\subfigure[The graph $G'$ defined in the proof of Lemma \ref{lem:x+}, where for clarity, agent 1 and her edges are grayed out. The black edges signify that both directions are present, i.e., they correspond to pairs of anti-parallel edges. On the right we show the actual allocation resulting from $G$.]{
\mbox{\includegraphics[scale = 1]{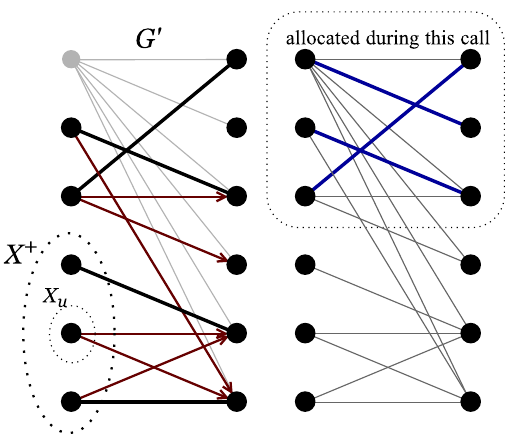}}
\label{subfig:G\rq{}}
}
\caption{Ilustration of $G$, $G'$ and $X^+$.}\label{fig:X+}
\end{figure}

We note that the analysis of the algorithm is tight, given the analysis on $\rho_n$ (see Section 3.3 of~\citet{PW14}). 
Improving further on the approximation ratio of $2/3$ seems to require drastically new ideas and it is a challenging open problem. We stress that even a PTAS is not currently ruled out by the lower bound constructions~\citep{KPW16,PW14}. 
Related to this, in the next section we consider two special cases in which we can obtain better positive results.

\section{Two Special Cases}
\label{app-sec:special}

In this section, we consider two interesting special cases, where we  have improved approximations.
The first is the case of $n=3$ agents, where we obtain a $7/8$-approximation, improving on the $3/4$-approximation of~\citet{PW14}. 
The second is the case where all values for the goods belong to $\{0, 1, 2\}$. This is an extension of the $\{0, 1\}$ setting discussed in~\citet{BL16}, and we show how to get an exact allocation without any approximation loss.  

\subsection{The Case of $n=3$ Agents}
\label{app-sec:n=3}

For $n=2$, it is pointed out in~\citet{BL16} that maximin share allocations exist via an analog of the cut and choose protocol.
Using the PTAS of~\citet{Woeginger97}, we can then have a ($1-\varepsilon$)-approximation in polynomial time. 
In contrast, as soon as we move to $n=3$, things become more interesting. It is proven that with $3$ agents there exist instances where no maximin share allocation exists~\citep{PW14}. The best known approximation guarantee is $\frac{3}{4}$ by 
observing that the quantity $\rho_n$, defined in Section~\ref{sec:2/3}, satisfies $\rho_3\geq \frac{3}{4}$. 

We provide a different algorithm, improving the approximation to $\frac{7}{8}-\varepsilon$.
To do this, we combine ideas from both algorithms presented so far in Sections~\ref{sec:1/2} and~\ref{sec:2/3}. 
The main result of this subsection is as follows:

\begin{theorem}\label{thm:7/8}
	Let $N=\{1, 2, 3\}$ be a set of three agents with additive valuations, and let $M$ be a set of goods. 
	For any constant $\varepsilon>0$, Algorithm \ref{fig:alg-n=3}
	produces in polynomial time  an allocation $(S_1,S_2,S_3)$,
	such that 
	\[ v_i(S_i) \geq \left(\frac{7}{8}-\varepsilon\right) \mms_i(3, M)\,,\ \forall i\in N\,.\]
\end{theorem}

The algorithm is shown below. Before we prove
Theorem \ref{thm:7/8}, we provide here a brief outline of how the algorithm works.\medskip

\noindent {\bf Algorithm Outline:} 
First, approximate values for the $\mms_i$s are calculated as before. Then, if there are items with large value to some agent, in analogy to 
Algorithm~\ref{fig:alg-1/2}, we first allocate
one of those reducing this way the problem to the simple case of $n=2$. 
If there are no items of large value, then  the first agent partitions the items as in Algorithm \ref{fig:alg-pw_rec}. In the case where this partition does not satisfy all three agents, then the second agent repartitions two of the bundles of the first agent. Actually, she tries two different such repartitions, and we show that at least one of them works out.
The definition of a bipartite preference graph and a corresponding matching (as in Algorithm \ref{fig:alg-pw_rec}) is never mentioned explicitly here.
However, the main idea (and the difference with Algorithm \ref{fig:alg-pw_rec}) is that if there are several  ways to pick a perfect matching between $X\mysetminus X^+$ and a subset of 
$Y\mysetminus \Gamma(X^+)$, then we try them all and choose the best one. Of course, since $n=3$, if there is no perfect matching
in the preference graph, then $X\mysetminus X^+$ is going to be just a single vertex, and we only have to examine two possible perfect 
matchings between $X\mysetminus X^+$ and a subset of $Y\mysetminus \Gamma(X^+)$. 

\begin{algorithm}
	\DontPrintSemicolon 
	{\small
		$\varepsilon'=\frac{8}{7}\varepsilon$ \;
		Compute a $(1-\varepsilon)$-approximation $\xi_i$ of $\mms_i(3, M)$ for $i\in \{1, 2, 3\}$. \;
		\eIf {$\exists i\in \{1, 2, 3\}, j\in M$ such that $v_{ij}\ge \frac{7}{8} \xi_i$}   {  \label{line:1st_if_1st} 
			Give item $j$ to agent $i$ and divide $M\mysetminus \{j\}$ among the other two agents in a ``cut-and-choose" fashion.\;	\label{line:cut-n-choose}}
		{
			Agent $1$ computes a $(1-\varepsilon)$-approximate maximin partition of $M$ into three sets, say $(A_1, A_2, A_3)$.   \; \label{line:1st_else_1st}
			\eIf {$\exists j_2, j_3\in \{1, 2, 3\}$ such that $j_2\neq j_3$, $v_2(A_{j_2})\ge \frac{7}{8} \xi_2$ 
				and $v_3(A_{j_3})\ge \frac{7}{8} \xi_3$}    { \label{line:2nd_if_1st}
				Give set $A_{j_2}$ to agent $2$, set $A_{j_3}$ to agent $3$, and the last set to agent $1$.  \;  \label{line:2nd_if_last} }
			{
				There are two sets that have value less than $\frac{7}{8}\xi_2$ w.r.t.~agent 2, say for simplicity $A_2$ and $A_3$. \; \label{line:main_1st}
				Agent $2$ computes $(1-\varepsilon')$-approximate 2-maximin partitions of $A_1\cup A_2$ and $A_1\cup A_3$, say $(B_1, B_2)$ and $(B_1', B_2')$ respectively, and discards the partition with the smallest maximin value. Let $(D_1, D_2)$ be the partition she keeps.\;  \label{line:discards}
				Agent $3$ takes the set she prefers from $(D_1, D_2)$; agent $2$ gets the other, and agent $1$ gets $M\mysetminus(D_1\cup D_2)$.  \; \label{line:main_last}
			}}
		}
		\caption{$\textsc{apx-3-mms}(\varepsilon, M, v_1, v_2, v_3)$}\label{fig:alg-n=3}
	\end{algorithm}

	\begin{proof}[Proof of Theorem \ref{thm:7/8}]
		First, note that for constant $\varepsilon$ the algorithm runs in time polynomial in $|M|$. 
		Next, we prove the correctness of the algorithm.
		
		If the output is computed in lines \ref{line:1st_if_1st}-\ref{line:cut-n-choose} then for agent $i$, as defined in line~\ref{line:1st_if_1st}, the value 
		she receives is at least
		$\frac{7}{8} \xi_i \ge \frac{7}{8}(1-\varepsilon)\mms_i(3, M) > \big(\frac{7}{8}-\varepsilon\big)\mms_i(3, M)$.  
		The remaining two agents $i_1, i_2$ essentially apply an approximate version of a cut and choose protocol. Agent $i_1$ computes a 
		$(1-\varepsilon)$-approximate 2-maximin partition of $M\mysetminus \{j\}$, say $(C_1, C_2)$, then agent $i_2$ takes the set she prefers among $C_1$ and $C_2$, and agent $i_1$ gets the other.
		By the monotonicity lemma (Lemma \ref{lem:monotonicity}), we know that $\mms_{i_1}(2, M\mysetminus\{j\}) \ge \mms_{i_1}(3, M)$, and thus 
		no matter which set is left for agent ${i_1}$, she is guaranteed a total value of at least 
		$(1-\varepsilon)\mms_{i_1}(3, M) > \big(\frac{7}{8}-\varepsilon\big)\mms_{i_1}(3, M)$. 
		Similarly, we have $\mms_{i_2}(2, M\mysetminus\{j\}) \ge \mms_{i_2}(3, M)$, and therefore $v_{i_2}(M\mysetminus\{j\})\ge 2\mms_{i_2}(3, M)$.
		Since $i_2$ chooses before $i_1$, she is guaranteed a total value that is at least $\mms_{i_2}(3, M) > \big(\frac{7}{8}-\varepsilon\big)\mms_{i_2}(3, M)$.
		
		If the output is computed in lines \ref{line:1st_else_1st}-\ref{line:2nd_if_last} then clearly all agents receive a ($7/8-\varepsilon$)-approximation, 
		since for agent $1$ it does not matter which of the $A_i$s she gets. 
		
		The most challenging case is when the output is computed in lines \ref{line:main_1st}-\ref{line:main_last} (starting with the partition from line \ref{line:1st_else_1st}). Then, as before, agent $1$ 
		receives a value that is at least a ($7/8-\varepsilon$)-approximation no matter which of the three sets she gets. For agents 2 and 3, 
		however, the analysis is not straightforward. We need the following lemma. 
		
		\begin{lemma}\label{lem:7/8}
			Let $N, M, \varepsilon$ be as above, such that for all $j\in M$ we have $v_{2j}<\frac{7}{8} \xi_2$.
			Consider any partition of $M$ into 3 sets $A_1, A_2, A_3$ and assume that there are no $j_2, j_3\in \{1, 2, 3\}$ such that $j_2\neq j_3$,  
			$v_2(A_{j_2})\ge \frac{7}{8} \xi_2$ and $v_3(A_{j_3})\ge \frac{7}{8} \xi_3$. Then lines \ref{line:main_1st}-\ref{line:main_last} of 
			Algorithm \ref{fig:alg-n=3} produce an allocation $(S_2,S_3)$ for agents 2 and 3, 
			such that for $i\in \{2, 3\}$:
			$ v_i(S_i) \geq \left(\frac{7}{8}-\varepsilon\right) \mms_i(3, M).$
			Moreover, if agent 1 is given set $A_k$, then $S_2\cup S_3 = \bigcup_{\ell\in N\mysetminus k} A_{\ell}$.
		\end{lemma}
		
		Clearly, Lemma \ref{lem:7/8} completes the proof.
	\end{proof}
	
	Before stating the proof of Lemma \ref{lem:7/8}, we should mention how it is possible to go beyond the previously known $\frac{3}{4}$-approximation. 
	As noted above, $\rho_n$ is by definition the best guarantee we can get, based only on the fact that the complement of the set left to be shared is not too large.
	As a result, the $\frac{7}{8}$ ratio cannot be guaranteed 
	just by the excess value. Instead, in addition to making sure that the remaining items are valuable enough for the remaining agents, we further argue about how a maximin partition would distribute those items.
%
%

There is an alternative interpretation of Algorithm \ref{fig:alg-n=3} in terms of Algorithm \ref{fig:alg-pw}. Whenever only a single agent (i.e., agent 1) is going to become satisfied in the first recursive call, we try all possible maximum matchings of the graph $G$ for the calculation of $X^+$. Then we proceed with the ``best'' such matching. Here, for $n=3$, this means we only have to consider two possibilities for the set agent 1 is going to get matched to; it is either $A_2$ or $A_3$ (subject to the assumptions in Algorithm \ref{fig:alg-n=3}).

	\begin{proof}[Proof of Lemma \ref{lem:7/8}]
		First, recall that $v_2(M) \ge 3\mms_2(3, M) \ge 3\xi_2$. 
		Like in the description of the algorithm we may assume that agent 1 gets set $A_3$, without loss of generality. 
		Before we move to the analysis we should lay down some facts. Let $(B_1, B_2)$ be agent 2's $(1-\varepsilon')$-approximate maximin 
		partition of $A_1\cup A_2$ computed in line \ref{line:discards}; similarly $(B'_1, B'_2)$ is agent 2's $(1-\varepsilon')$-approximate maximin 
		partition of $A_1\cup A_3$. We may assume that $v_2(B_1)\ge v_2(B_2)$. Also, assume that in line \ref{line:discards} of the algorithm we have  
		$(D_1, D_2)=(B_1, B_2)$, i.e., $\min\{v_2(B_1'), v_2(B_2')\}\le v_2(B_2)$ and $M\mysetminus(D_1\cup D_2)=A_3$. The case where $(D_1, D_2)=(B_1', B_2')$ is symmetric. Our goal is to show that $v_2(B_2) \ge \left( \frac{7}{8}-\varepsilon\right) \mms_2(3, M)$.
		For simplicity, we write $\mms_2$ instead of $\mms_2(3, M)$.
		
		Note, towards a contradiction, that 
\begin{IEEEeqnarray*}{rCl}
			v_2(B_2) & < & \left(\frac{7}{8}-\varepsilon\right)\mms_2 \Rightarrow \\
			(1-\varepsilon')\mms_2(2, A_1\cup A_2) & < & \left(\frac{7}{8}-\varepsilon\right)\mms_2\Rightarrow \\
			(1-\varepsilon')\mms_2(2, A_1\cup A_2) & < & \left(\frac{7}{8}-\frac{7}{8}\varepsilon'\right)\mms_2\Rightarrow \\
			\mms_2(2, A_1\cup A_2) & < & \frac{7}{8}\mms_2\,.
		\end{IEEEeqnarray*}
		
		Moreover, this means $\min\{v_2(B_1'), v_2(B_2')\} < \left(\frac{7}{8}-\varepsilon\right)\mms_2$ as well, which leads to
		$\mms_2(2, A_1\cup A_3)  <  \frac{7}{8}\mms_2$.  So, it suffices to show that either
		$\mms_2(2, A_1\cup A_2)$ or $\mms_2(2, A_1\cup A_3)$ is at least $\frac{7}{8}\mms_2$.
		This statement is independent of the $B_i$s and in what follows we consider exact maximin partitions with respect to agent 2.
		Before we proceed, we should 
		make clear that for the case we are analyzing there are indeed exactly two sets in $\{A_1, A_2, A_3\}$ each with value less than $\frac{7}{8}\mms_2$ with 
		respect to agent 2, as claimed in line \ref{line:main_1st} of the algorithm. Indeed, notice that in any partition of $M$ there is at least one set with value 
		at least $\mms_2$ with 
		respect to agent $2$, due to the fact that $v_2(M)\ge 3 \mms_2$ and
		by the definition of a maximin partition. If, however, there were at least 2 sets in $\{A_1, A_2, A_3\}$ with value at least
		$\frac{7}{8} \xi_2$, then we would be at the case handled in steps \ref{line:1st_else_1st}-\ref{line:2nd_if_last}.
		Hence, there will be exactly two sets each with value less than 
		$\frac{7}{8}\xi_2\le\frac{7}{8}\mms_2$ for agent 2 and as stated in the algorithm we assume these are the sets $A_2, A_3$.
		
		Consider a $3$-maximin share allocation $(A_1', A_2', A_3')$ of $M$ with respect to agent 2. Let $F_i=A_i'\cap A_3$ for $i=1, 2, 3$.
		Without loss of generality, we may assume that $v_2(F_1)\le v_2(F_2)\le v_2(F_3)$. 
		
		If $v_2(F_1)\le \frac{1}{8}\mms_2$, then 
		the partition $(A_1'\mysetminus A_3, (A_2'\cup A_3')\mysetminus A_3)$ is a partition of $A_1\cup A_2$ such that 
		\[ v_2(A_1'\mysetminus A_3)= v_2(A_1')-v_2(F_1)\ge \mms_2 - \frac{1}{8}\mms_2 = \frac{7}{8}\mms_2\]
		and
		\[ v_2((A_2'\cup A_3')\mysetminus A_3) \ge v_2(A_2') + v_2(A_3') - v_2(A_3)\ge 2\mms_2 - \frac{7}{8}\mms_2 = \frac{9}{8}\mms_2\,. \] 
		So, in this case we conclude that $\mms_2(2, A_1\cup A_2)\ge \frac{7}{8}\mms_2$.
		
		On the other hand, if $v_2(F_1) > \frac{1}{8}\mms_2$ we are going to show that $\mms_2(2, A_1\cup A_3)\ge \frac{7}{8}\mms_2$.
		Towards this we consider a $2$-maximin share allocation $(C_1, C_2)$ of $A_1$ with respect to agent 2 and 
		let us assume that $v_2(C_1)\ge v_2(C_2)$. For a rough depiction of the different sets involved in the following arguments, see Figure \ref{fig:sec5-fig1}.

\begin{figure}[h]
\centering
\mbox{\includegraphics[scale = 0.45]{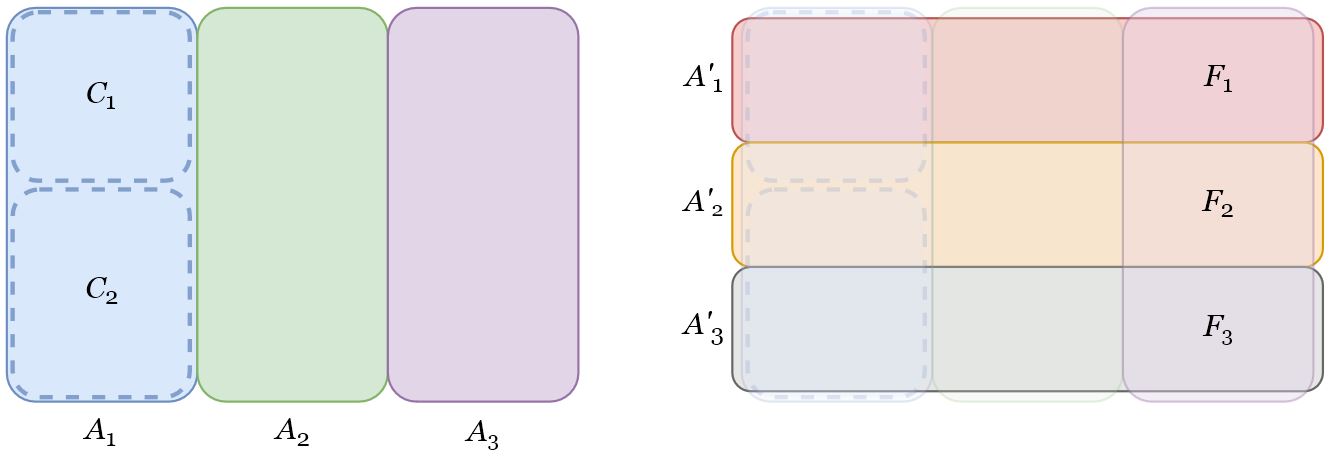}}
\caption{Assuming that the set of items $M$ is represented by a rectangle, here is a depiction of several sets involved in the proof of Lemma \ref{lem:7/8}. Recall that $(A_1, A_2, A_3)$ and $(A_1', A_2', A_3')$ are partitions of $M$, $(C_1, C_2)$ is a partition of $A_1$, and  $F_i=A_i'\cap A_3$ for $i=1, 2, 3$.} \label{fig:sec5-fig1}
\end{figure}

\begin{myclaim}\label{cl:sums7/8}
For  $C_1, C_2, A_3, F_1, F_2, F_3$ as above, we have
\begin{itemize}
\item[(i)] $v_2(A_3)+v_2(C_2) \ge \frac{7}{8}\mms_2$, and
\item[(ii)] $v_2(F_1) + v_2(F_2) + v_2(C_1) > \frac{7}{8}\mms_2$. 
\end{itemize}
\end{myclaim}
		\begin{proof} \renewcommand{\qedsymbol}{{\footnotesize $\boxdot$}}
			Note that 
			\[v_2(C_1)+v_2(C_2)+v_2(A_3) = v_2(M)-v_2(A_2) > 3\mms_2 - \frac{7}{8}\mms_2 = \frac{17}{8}\mms_2\,.\]
			If $v_2(A_3)+v_2(C_2) < \frac{7}{8}\mms_2$ then $v_2(C_1)> \frac{10}{8}\mms_2$.  Moreover, 
			\[v_2(A_3)=v_2(F_1)+v_2(F_2)+v_2(F_3)\ge 3v_2(F_1)>\frac{3}{8}\mms_2\,,\]
			so $v_2(A_3)+v_2(C_2) < \frac{7}{8}\mms_2$ implies that $v_2(C_2)< \frac{4}{8}\mms_2$.
			
			Let $d$ denote the difference $v_2(C_1)-v_2(C_2)$; clearly $d >\frac{6}{8}\mms_2$.
			It is not hard to see that $\min_{j\in C_1}v_{2j}\ge d$. Indeed, suppose there existed some $ j\in C_1$ such that $v_{2j} < d$. Then, by moving $j$ from $C_1$ to $C_2$ we increase the minimum value of the partition, which contradicts the choice of $(C_1, C_2)$. 

Since $v_2(C_1)> \frac{10}{8}\mms_2$ and no item has value more than $\frac{7}{8}\mms_2$ for agent 2, this means that $C_1$ contains at least two items. Thus, $v_2(C_1) \ge \min_{j\in C_1}v_{2j} > \frac{12}{8}\mms_2$.
			
			Now, for any item $g\in \argmin_{j\in C_1}v_{2j}$, the partition $(\{g\}, A_1\mysetminus \{g\})$ is strictly better than $(C_1, C_2)$, since 
			$v_{2g} > \frac{6}{8}\mms_2 > v_2(C_2)$ and $v_2(A_1\mysetminus \{g\}) = v_2(A_1)-v_{2g} \ge v_2(C_1)-v_{2g} > \frac{12}{8}\mms_2-\frac{6}{8}\mms_2 = \frac{6}{8}\mms_2 > v_2(C_2)$.
			Again, this contradicts the choice of $(C_1, C_2)$. 
			Hence, it must be that $v_2(A_3)+v_2(C_2) \ge \frac{7}{8}\mms_2$.\medskip

			
			The proof of {\it (ii)} is simpler. Notice that 
\begin{IEEEeqnarray*}{rCl}
v_2(F_1) + v_2(F_2) + v_2(C_1)& \ge & v_2(F_1) + v_2(F_1) + \frac{1}{2}v_2(A_1) \\
 & > & \frac{1}{8}\mms_2 +\frac{1}{8}\mms_2 +  \frac{1}{2}\big(3\mms_2 - \frac{7}{8}\mms_2 - \frac{7}{8}\mms_2\big) = \frac{7}{8}\mms_2\,.\qedhere
		\end{IEEEeqnarray*}
		\end{proof}
		
		Now, if $v_2(C_1) \ge \frac{7}{8}\mms_2$ then {\it (i)} of Claim \ref{cl:sums7/8} implies that 
		$\min \{v_2(C_1), v_2(A_3\cup C_2)\} \ge \frac{7}{8}\mms_2$. Similarly, if $v_2(F_3) + v_2(C_2) \ge \frac{7}{8}\mms_2$ 
		then {\it (ii)} of Claim \ref{cl:sums7/8} implies that 
		$\min \{v_2(F_1\cup F_2\cup C_1), v_2(F_3\cup C_2)\} \ge \frac{7}{8}\mms_2$. In both cases, we have 
		$\mms_2(2, A_1\cup A_3)\ge \frac{7}{8}\mms_2$. So, it is left to examine the case where both $v_2(C_1)$
		and $v_2(F_3) + v_2(C_2)$ are less than $\frac{7}{8}\mms_2$. 
		
		\begin{myclaim}\label{cl:sums7/8_second}
			Let $C_1, C_2, A_3, F_1, F_2, F_3$ be as above and $\max\{v_2(C_1), v_2(F_3\cup C_2)\} < \frac{7}{8}\mms_2$. Then 
			$\min \{v_2(F_1\cup C_1), v_2(F_2\cup F_3\cup C_2)\} \ge \frac{7}{8}\mms_2$.
		\end{myclaim}
		\begin{proof} \renewcommand{\qedsymbol}{{\footnotesize $\boxdot$}}
			Recall that $v_2(A_1)+v_2(A_3)>\frac{17}{8}\mms_2$. Suppose $v_2(F_1\cup C_1)< \frac{7}{8}\mms_2$. Then $v_2(F_2\cup F_3\cup C_2)> \frac{10}{8}\mms_2$. Since $v_2(F_3\cup C_2)< \frac{7}{8}\mms_2$ we have $v_2(F_2)> \frac{3}{8}\mms_2$. But then we get the contradiction
			\[ \frac{7}{8}\mms_2 > v_2(A_3)= v_2(F_1) + v_2(F_2) + v_2(F_3) \ge \frac{1}{8}\mms_2 + \frac{3}{8}\mms_2 + \frac{3}{8}\mms_2 = \frac{7}{8}\mms_2\,.\]
			Hence, $v_2(F_1\cup C_1) \ge \frac{7}{8}\mms_2$. Similarly, suppose $v_2(F_2\cup F_3\cup C_2)< \frac{7}{8}\mms_2$. Then $v_2(F_1\cup C_1)> \frac{10}{8}\mms_2$. Since $v_2(C_1)< \frac{7}{8}\mms_2$ we have $v_2(F_1)> \frac{3}{8}\mms_2$. Then we get the contradiction
			\[ \frac{7}{8}\mms_2 > v_2(A_3)= v_2(F_1) + v_2(F_2) + v_2(F_3) \ge \frac{3}{8}\mms_2 + \frac{3}{8}\mms_2 + \frac{3}{8}\mms_2 = \frac{9}{8}\mms_2\,.\]
			Hence, $v_2(F_2\cup F_3\cup C_2) \ge \frac{7}{8}\mms_2$. 			
		\end{proof}
		
		Claim \ref{cl:sums7/8_second} implies $\mms_2(2, A_1\cup A_3)\ge \frac{7}{8}\mms_2$ and this concludes the proof.\qedhere
	\end{proof}


\subsection{Values in $\{0, 1, 2\}$}
\label{subsec:012}
\citet{BL16} consider a binary setting where all valuation functions take values in $\{0, 1\}$, i.e.,
for each $i\in N$, and $j\in M$, $v_{ij}\in \{0, 1\}$. This can correspond to expressing approval or disapproval for each item.
It is then shown that it is always possible to find a maximin share allocation in polynomial time. In fact, they show that the Greedy Round-Robin algorithm, presented in Section \ref{sec:1/2}, computes such an allocation in this case.  

Here, we extend this result to the setting where each $v_{ij}$ is in $\{0, 1, 2\}$, allowing the agents to express two types of approval for the items.
Enlarging the set of possible values from $\{0, 1\}$ to $\{0, 1, 2\}$ by just one extra possible value 
makes the problem significantly more complex. Greedy Round-Robin does not work in this case, so a different algorithm is developed.

\begin{theorem}
\label{thm:012}
Let $N=[n]$ be a set of agents and $M=[m]$ be a set of items. If for any $i\in N$, agent $i$ has 
a valuation function $v_i$ such that $v_{ij}\in\{0, 1, 2\}$ for any $j\in M$, then we can find, 
in time $O(nm \log m)$, an allocation $(T_1, \ldots, T_n)$ of $M$ so that $v_i(T_i)\ge \mms_i(n, M)$ 
for every $i\in [n]$.
\end{theorem}

To design our algorithm, we make use of an important observation by~\citet{BL16} that allows us to reduce 
appropriately the space of valuation functions that we are interested in. We say that the agents have 
{\it fully correlated} valuation functions if they agree on a common ranking of the items in decreasing 
order of values. That is, $\forall i \in N$, if $M = \{1, 2, ...,m\}$, we have $v_{i1} \ge v_{i2} \ge \ldots \ge v_{im}$. 
In \citet{BL16}, the authors show that to find a maximin share allocation for any set of  
valuation functions, it suffices to do so in an instance where the valuation functions are fully 
correlated. This family of instances seems to be the difficulty in computing such allocations.
Actually, 
their result preserves approximation ratios as well (with the same proof); hence we state 
this stronger version. 
For a valuation function $v_i$ let $\sigma_i$ be a permutation on the items such that 
$v_i(\sigma_i(j)) \ge v_i(\sigma_i(j+1))$ for $j\in \{1, \ldots, m-1\}$. We denote  the function 
$v_i(\sigma_i(\cdot))$ by $v_i^{_\uparrow}$. 
Note that $v_1^{_\uparrow}, v_2^{_\uparrow}, \ldots, v_n^{_\uparrow}$
are now fully correlated.

\begin{theorem}[\normalfont{\citep{BL16}}]
\label{thm:fully-corr}
Let $N=[n]$ be a set of agents with additive valuation functions, $M=[m]$ be a set of goods and $\rho\in(0, 1]$.
Given an allocation $(T_1, \ldots, T_n)$ of $M$ so that 
$v_i^{_\uparrow}(T_i)\ge \rho \mms_i(n, M)$ 
for every $i$, one can produce in linear time an allocation $(T'_1, \ldots, T'_n)$ of $M$ so that 
$v_i(T'_i)\ge \rho \mms_i(n, M)$ for every $i$.
\end{theorem}

We are ready to state a high level description of our algorithm. The detailed description, however, is deferred to the end of this subsection. The reason for this is that the terminology needed is gradually introduced through a series of lemmas motivating the idea behind the algorithm and proving its correctness. In fact, the remainder of the subsection is the proof of Theorem \ref{thm:012}. Algorithm \ref{fig:alg-012} in the end summarizes all the steps.   \medskip

\noindent {\bf Algorithm Outline:} We first construct $v_1^{_\uparrow}, v_2^{_\uparrow}, \ldots, v_n^{_\uparrow}$ and work with them instead. 
The Greedy Round-Robin algorithm may not directly work, but we partition the items in a similar fashion, although without giving them 
to the agents. Then, we show that it is possible to choose  some subsets of items and redistribute them in a way that guarantees that 
everyone can get a bundle of items with enough value. At a higher level, we could say that the algorithm simulates a variant of the 
Greedy Round-Robin, where for an appropriately selected set of rounds the agents choose in the reverse order. Finally, a maximin 
share allocation can be obtained for the original $v_i$s,  as described in \citet{BL16}. \smallskip


\begin{proof}[Proof of Theorem \ref{thm:012}]
According to Theorem \ref{thm:fully-corr} it suffices to focus on instances where
the valuation functions take values in $\{0, 1, 2\}$ and are fully correlated. Given such an instance
we distribute the $m$ objects into $n$ buckets in decreasing order, i.e., bucket $i$ will get items 
$i, n+i, 2n+i, \ldots$. Notice that this is compatible with how the Greedy Round-Robin algorithm could
distribute the items; however, we do not assign any buckets to any agents yet. We may assume that 
$m=kn$ for some  $k\in \mathbb{N}$; if not, we just add a few extra items with 0 value to everyone. 
It is convenient to picture the collection of buckets as the matrix
\[B ={ 
 \begin{pmatrix}
  (k-1)n+1 & (k-1)n+2 & \cdots & kn \\
  \vdots  & \vdots  & \ddots & \vdots  \\
  n+1 & n+2 & \cdots & 2n \\
  1 & 2 & \cdots & n
 \end{pmatrix}}\,,\] 
since our algorithm will systematically redistribute groups of items corresponding to rows of $B$. 

Before we state the algorithm, we establish some properties regarding these buckets and the way each 
agent views the values of these bundles.
First, we introduce some terminology. 
\begin{definition}
We say that agent $i$ is 
\begin{itemize}
\item \emph{satisfied} with respect to the current buckets, if all the buckets have 
value at least $\mms_i(n, M)$ according to $v_i$. 
\item \emph{left-satisfied} with respect to the current buckets, if she is not satisfied, but at least the $n/2$ leftmost buckets have 
value at least $\mms_i(n, M)$ according to $v_i$. 
\item \emph{right-satisfied} if the same as above hold, but for the rightmost $n/2$ buckets.
\end{itemize}
\end{definition} 

Now suppose that we see agent $i$'s view of the values in the buckets. A typical view would have the following form (recall the goods are ranked from highest to lowest value):
\[{\small
 \begin{pmatrix}
  0 & 0 & 0 & 0 & 0 & 0 &\cdots & 0 & 0 & 0 \\
  \cdot & \cdot & \cdot & \cdot & \cdot & \cdot &\cdots & \cdot & \cdot & \cdot \\
  1 & 1 & 1 & 1 & 1 & 0 &\cdots & 0 & 0 & 0 \\
  \cdot & \cdot & \cdot & \cdot & \cdot & \cdot &\cdots & \cdot & \cdot & \cdot \\
  1 & 1 & 1 & 1 & 1 & 1 &\cdots & 1 & 1 & 1 \\
  2 & 2 & 2 & 1 & 1 & 1 &\cdots & 1 & 1 & 1 \\
  \cdot & \cdot & \cdot & \cdot & \cdot & \cdot &\cdots & \cdot & \cdot & \cdot \\
  2 & 2 & 2 & 2 & 2 & 2 &\cdots & 2 & 2 & 2 
 \end{pmatrix}}\, \]
A row that has only $2$s for $i$ will be called a $\scriptstyle 2$-row for $i$. A row that has both $2$s 
and $1$s will be called a $\scriptstyle 2/1$-row for $i$, and so forth. An agent can also have a $\scriptstyle 2/1/0$-row. It is not necessary, of course, 
that an agent will have all possible types of rows in her view. Note, however, that there can be
at most one $\scriptstyle 1/0$-row and at most one $\scriptstyle 2/1$-row in her view. 
We first prove the following lemma for agents that are not initially satisfied. 

\begin{lemma}
\label{lem:left-satisfied}
Any agent not satisfied with respect to the initial buckets must have both a $\scriptstyle 1/0$-row and a 
$\scriptstyle 2/1$-row in her view of $B$. Moreover, initially all agents are either satisfied or 
left-satisfied.
\end{lemma}

\begin{proof} \renewcommand{\qedsymbol}{{\footnotesize $\boxdot$}}
Let us focus on the multiset of values of an agent that is not satisfied, say $i$. 
It is straightforward to see that if $i$ has no $1$s, or the number of $2$s is a multiple of $n$
(including 0), then agent $i$ gets value $\mms_i(n, M)$ from any bucket. So, $i$ must have a row with
both $2$s and $1$s. If this is a $\scriptstyle 2/1/0$-row, then again it is easy to see that the initial
allocation is a maximin share allocation for $i$. So, $i$ has a $\scriptstyle 2/1$-row. The only case
where she does not have a $\scriptstyle 1/0$-row is if the total number of $1$s and $2$s is a multiple
of $n$.But then the maximum and the minimum value of the initial buckets differ by 1, hence we have a
maximin share allocation and $i$ is satisfied.

Next we show that an agent $i$ who is not initially satisfied is left-satisfied. 
In what follows we only refer to $i$'s view. 
Buckets $B_{1}$ and $B_{n}$, indexed by the corresponding columns of $B$, have 
maximum and minimum total value respectively. Since $i$ is not
satisfied, we have $v_i(B_{1}) \ge v_i(B_{n}) + 2$, but the way we distributed the items guarantees that the difference between
any two buckets is at most the largest value of an item; so $v_i(B_{1}) = v_i(B_{n}) + 2$.
Moreover, since $v_i(M)\ge n \mms_i(n, M)$ and $v_i(B_{n})<\mms_i(n, M)$, we must have $v_i(B_{1})>\mms_i(n, M)$ .
This implies that $v_i(B_{1})=\mms_i(n, M) +1$ and $v_i(B_{n})=\mms_i(n, M) -1$.

More generally, we have buckets of value $\mms_i(n, M) +1$ (leftmost columns), we have buckets of value 
$\mms_i(n, M) -1$ (rightmost columns), and maybe some other buckets of value $\mms_i(n, M)$ (columns 
in the middle). We know that the total value of all the items is at least $n\mms_i(n, M)$, so, by 
summing up the values of the buckets, we conclude that there must be at most $n/2$ buckets of value 
$\mms_i(n, M) -1$. Therefore $i$ is left-satisfied.
\end{proof}

So far we may have some agents that could take any bucket and some agents that would take any of the 
$n/2$ (at least) first buckets. Clearly, if the left-satisfied agents are at most $n/2$ then we can
easily find a maximin share allocation. However, there is no guarantee that there are not too many
left-satisfied agents initially, so we try to fix this by reversing some of the rows of $B$. 
To make this precise, we say that {\it we reverse the $i$th row of $B$} when we take items 
$(i-1)n+1, (i-1)n+2, \ldots, in$ and we put item $in$ in bucket 1, item $in-1$ in bucket 2, etc.

The algorithm then proceeds by picking a subset of rows of $B$ and reversing them. The rows are chosen appropriately
so that the resulting buckets (i.e., the columns of $B$) can be easily paired with the agents 
to get a maximin share allocation.
First, it is crucial to understand the effect that the reversal of a set of rows has to an agent. 

\begin{lemma}
\label{lem:reverse-satisfied}
Any agent satisfied with respect to the initial buckets remains satisfied independently of the rows 
of $B$ that we may reverse. On the other hand, any agent not satisfied with respect to the initial 
buckets, say agent $i$, is affected if we reverse her $\scriptstyle 1/0$-row or her 
$\scriptstyle 2/1$-row. If we reverse only one of those, then $i$ becomes satisfied with respect to 
the new buckets; if we reverse both, then $i$ becomes right-satisfied. The reversal of any other 
rows is irrelevant to agent $i$. 
\end{lemma}

\begin{proof} \renewcommand{\qedsymbol}{{\footnotesize $\boxdot$}}
Fix an agent $i$. First notice that, due to symmetry, reversing any row that for $i$ is a $\scriptstyle 2$-row, 
a $\scriptstyle 1$-row, or a $\scriptstyle 0$-row does not improve or
worsen the initial allocation from $i$'s point of view. Also, clearly, reversing both the 
$\scriptstyle 1/0$-row and the $\scriptstyle 2/1$-row of a left-satisfied agent makes her right-satisfied.
Similarly, if $i$ is satisfied and has a $\scriptstyle 2/1/0$-row, or has a $\scriptstyle 2/1$-row but
no $\scriptstyle 1/0$-row, or has a $\scriptstyle 1/0$-row but no $\scriptstyle 2/1$-row, then 
reversing those keeps $i$ satisfied. 

The interesting case is when $i$ has both a $\scriptstyle 1/0$-row and a $\scriptstyle 2/1$-row.
If $i$ is satisfied, then even removing her $\scriptstyle 1/0$-row leaves all the buckets with at least 
as much value as the last bucket; so reversing it keeps $i$ satisfied. A similar argument  holds for 
$i$'s  $\scriptstyle 2/1$-row as well.
If $i$ is not satisfied, then the difference of the values of the first and the last bucket will be 2. 
Like in the proof of Lemma \ref{lem:left-satisfied}, the number of columns that have
$1$ in $i$'s $\scriptstyle 1/0$-row and $2$ in $i$'s $\scriptstyle 2/1$-row (i.e., total value 
$\mms_i(n, M) +1$) are at least as many as the columns that have $0$ in $i$'s $\scriptstyle 1/0$-row 
and $1$ in $i$'s $\scriptstyle 2/1$-row (i.e., total value $\mms_i(n, M) -1$). So, by reversing her 
$\scriptstyle 1/0$-row, the values of all the ``worst'' (rightmost) buckets increase by 1, the values 
of some of the ``best'' (leftmost) buckets decrease by 1, and the values of the buckets in the middle 
either remain the same or increase by 1. The difference between the best and the worst buckets now 
is 1 (at most), so this is a maximin share allocation for $i$ and she becomes satisfied. Due to 
symmetry, the same holds for reversing $i$'s $\scriptstyle 2/1$-row only.
\end{proof}

Now, what Lemma \ref{lem:reverse-satisfied} guarantees is that when we reverse some of the rows of the 
initial $B$, we are left with agents that are either satisfied, left-satisfied, or right-satisfied. 
If the rows are chosen so that there are at most $n/2$ left-satisfied and at most $n/2$ right-satisfied
agents, then there is an obvious maximin share allocation: to any left-satisfied agent we 
arbitrarily give 
one of the first $n/2$ buckets, to any right-satisfied agent we arbitrarily give one of the last 
$n/2$ buckets, and to each of the remaining agents we arbitrarily give one of the remaining buckets.
In Lemma \ref{lem:easy-coloring} below, we prove that it is easy to find which rows to reverse to 
achieve that. 

We use a graph theoretic formulation of the problem for clarity. 
With respect to the initial buckets, we define a graph $G=(V, E)$ with $V=[k]$, i.e., $G$ has a
vertex for each row of $B$. Also, for each left-satisfied agent $i$, $G$ has an edge connecting
$i$'s $\scriptstyle 1/0$-row and $\scriptstyle 2/1$-row. We delete, if necessary, any multiple edges 
to get a simple graph with $n$ edges at most. We want to color the vertices of $G$ with two colors, 
``red'' (for reversed rows) and ``blue'' (for non reversed), so that the number of edges having both 
endpoints red is at most $n/2$ and at the same time the number of edges having both endpoints blue is 
at most $n/2$.
Note that if we reverse the rows that correspond to red vertices, then the agents with red endpoints
become right-satisfied, the agents with blue endpoints remain left-satisfied and the agents with both
colors become satisfied. Moreover, the initially satisfied agents are not affected, and we can find
a maximin share allocation as previously discussed. 
This is illustrated in Figure \ref{fig:alg-012-graph} below.

\begin{figure}[h]
\centering
\mbox{\includegraphics[scale = 0.85]{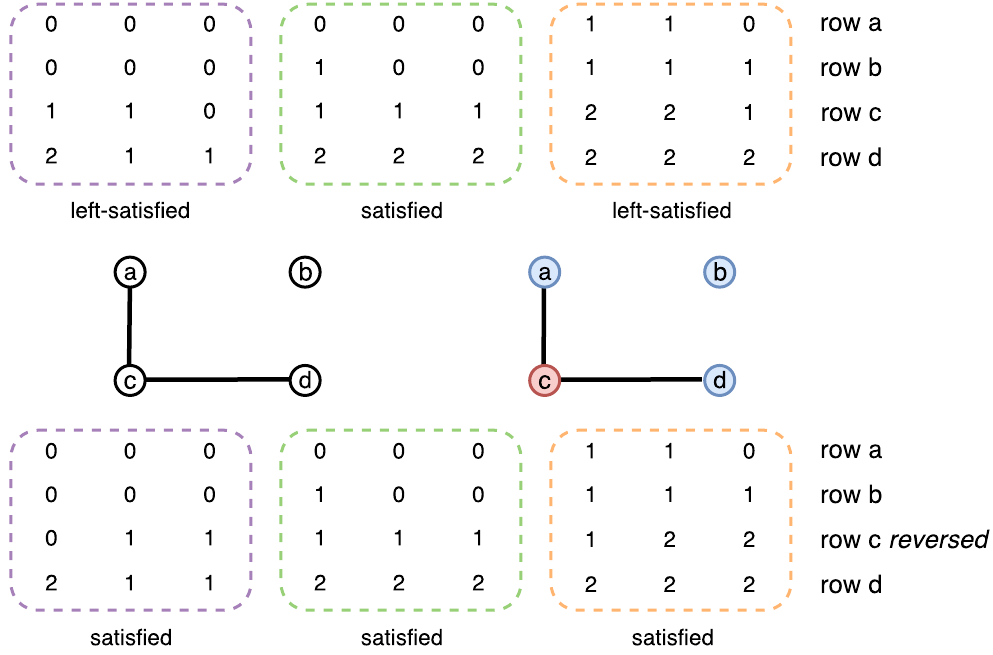}}
\caption{Assuming an instance with 3 agents and 11 items, the tables on top are the three different views on the initial buckets. This results in the graph shown in the middle---before and after the coloring. By reversing row $c$ that corresponds to a red vertex, every agent becomes satisfied and thus any matching of the columns to the agents defines an MMS allocation.} \label{fig:alg-012-graph}
\end{figure}

\begin{lemma}
\label{lem:easy-coloring}
Given graph $G$ defined above, in time $O(k+n)$ we can color the vertices with two colors, red and blue, 
so that the number of edges with two red endpoints is less than $n/2$ and the number of edges with two 
blue endpoints is at most $n/2$.
\end{lemma}

\begin{proof} \renewcommand{\qedsymbol}{{\footnotesize $\boxdot$}}
We start with all the vertices colored blue, and we arbitrarily recolor vertices red, one at a time, 
until the number of edges with two blue endpoints becomes at most $|E|/2$ for the first time. Assume this 
happens after recoloring vertex $u$. Before turning $u$ from blue to red, the number of edges with 
at most one blue endpoint was strictly less than $|E|/2$. Also, the recoloring of $u$ did not force any 
of the edges with two blue endpoints to become edges with two red endpoints. So, the number of edges 
with two red endpoints after the recoloring of $u$ is at most equal to the number of edges with at 
most one blue endpoint before the recoloring of $u$, i.e., less than $|E|/2$. To complete the proof,
notice that $|E|\le n$. For the running time, notice that each vertex changes color at most once
and when this happens we only need to examine the adjacent vertices in order to update the counters on
each type of edges (only red, only blue, or both).
\end{proof}

Lemma \ref{lem:easy-coloring} completes the proof of correctness for Algorithm \ref{fig:alg-012} that is summarized 
below. For the running time notice that $v_1^{_\uparrow}, \ldots, v_n^{_\uparrow}$ can be computed 
in $O(nm \log m)$, since we get $v_i^{_\uparrow}$ by sorting $v_{i1}, \ldots, v_{im}$. Also step 
\ref{line:left-sat} can be  computed in $O(nm)$; for each agent $i$ we scan the first column of $B$ 
to find her (possible) $\scriptstyle 1/0$-row  and $\scriptstyle 2/1$-row, and then in $O(n)$ we 
check whether she is left-satisfied by checking that the positions that have $1$ in $i$'s 
$\scriptstyle 1/0$-row and $2$ in $i$'s $\scriptstyle 2/1$-row are at least as many as the positions 
that have $0$ in $i$'s $\scriptstyle 1/0$-row and $1$ in $i$'s $\scriptstyle 2/1$-row.
\end{proof}

\begin{algorithm}[H]
\DontPrintSemicolon 
{\small
Let $k=\lceil\frac{m}{n}\rceil$. Add $k n -m$ dummy items with value 0 for everyone. \;
\If {$v_1, \ldots, v_n$ are not fully correlated} {
Compute $v_1^{_\uparrow}, \ldots, v_n^{_\uparrow}$ and use them instead. \;
}
Construct a $k\times n$ matrix $B$ so that $B_{ij}$ is the $(i-1)n+j$th item. \;
Find the set of left-satisfied agents and their corresponding $\scriptstyle 1/0$-rows 
			and $\scriptstyle 2/1$-rows. \;	\label{line:left-sat}
Construct a graph $G=([k], E)$ with $E=\{\{i,j\} | \exists \text{ left-satisfied  agent that }  
			 i \text{ and } j \text{ are her } \allowbreak {\scriptstyle 1/0}\text{-row and }
			 {\scriptstyle 2/1}\text{-row}\}$. \;
Color the vertices of $G$ with two colors, red and blue, so that the number of edges having both 
		endpoints red, and the number of edges having both endpoints blue, each is $\le n/2$. \;
Reverse the rows of $B$ that correspond to red vertices, and keep track of who is satisfied, 
		left-satisfied, or right-satisfied. \;
Arbitrarily give some of the first $n/2$ buckets (columns of $B$) to each of the left-satisfied agents 
		and some of the last $n/2$ buckets to each of the right-satisfied agents. Arbitrarily give the rest of 
		the buckets to the satisfied agents. \; \label{line:allocate_012}
\eIf {$v_1^{_\uparrow}, \ldots, v_n^{_\uparrow}$ were used} {
Based on the allocation in step \ref{line:allocate_012} compute and return a maximin share allocation for the original $v_i$s as described in \citep{BL16}. \;
}{Return the allocation in step \ref{line:allocate_012}.}
}
\caption{$\textsc{exact-mms}_{0,1,2}(N, M, V_N)$}\label{fig:alg-012}
\end{algorithm}

\section{A Probabilistic Analysis}
\label{sec:random}

As argued in the previous works~\citep{BL16,PW14}, it has been quite challenging to prove impossibility results.
Setting efficient computation aside, what is the best $\rho$ for which a $\rho$-approximate allocation does exist? All we know so far is that $\rho\neq 1$ by the elaborate constructions by~\citet{KPW16}, and \citet{PW14}. However, extensive experimentation
by~\citet{BL16} (and also by~\citet{PW14}), showed that in all generated instances, there always existed a maximin share allocation. 
Motivated by these experimental observations and by the lack of impossibility results, we present a probabilistic analysis, showing that indeed
we expect that in most cases there exist allocations where every agent receives her maximin share. In particular, we analyze the Greedy Round-Robin algorithm from Section \ref{sec:1/2} when each $v_{ij}$ is drawn from the uniform distribution over $[0,1]$. 

Recently, \citet{KPW16} show similar results  for a large set of distributions over $[0, 1]$,
including $U[0,1]$. Although, asymptotically, their results yield a theorem that is more general than ours, we consider our analysis to be of independent interest, since we have much better bounds on the probabilities for the special case of $U[0,1]$, even for relatively small values of $n$.

For completeness, before stating and proving our results, we include the version of Hoeffding's inequality we are going to use.

\begin{theorem}[\normalfont{\citep{Hoeffding63}}]\label{thm:Hoeffding}
	Let $X_1, X_2, \ldots, X_n$ be independent random variables with $X_i \in [0, 1]$ for $i\in[n]$. Then for the empirical mean $\bar{X} = \frac{1}{n} (X_1+\ldots +X_n)$ we have 
	$ \mathrm{P}\left( \bar{X}-\mathrm{E}[\bar{X}] \ge t\right)  \le \exp(-2nt^2) $.
\end{theorem}

We start with Theorem \ref{thm:random-prop}. Its proof is based on tools like Hoeffding\rq{}s and Chebyshev\rq{}s inequalities, and on a  careful estimation of the probabilities 
when $m<3n$. 
Note that for $m\geq 2n$, the theorem provides an even stronger guarantee than the maximin share (by Claim \ref{cl:upper_bound_1}).

\begin{theorem}\label{thm:random-prop}
	Let $N=[n]$ be a set of agents and $M=[m]$ be a set of goods, and assume that the $v_{ij}$s are i.i.d.\ random variables 
	that follow $U[0,1]$. Then, for $m\ge 2n$ and large enough $n$, the Greedy Round-Robin algorithm allocates to each agent $i$ 
	a set of goods of total value at least $\frac{1}{n}\sum_{j=1}^m v_{ij}$ with probability $1-o(1)$. 
	The $o(1)$ term is $O(1/n)$ when $m>2n$ and $O(\log n / n)$ when $m = 2n$.
\end{theorem}

\begin{proof} 
In what follows we assume that agent $1$ chooses first, agent $2$ chooses second, and so forth. 
We consider several cases for the different ranges of $m$. We first assume that $2n\le m <3n$. 

It is illustrative to consider the case of $m=2n$ and examine the $n$th agent that chooses last. Like all the agents in this case, she receives exactly two items; let $Y_n$ be the total value of those items.  From her perspective, she sees $n+1$ 
values chosen uniformly from $[0,1]$, picks the maximum of those, then u.a.r.~$n-1$ of the rest are removed, and she takes the last one 
as well. 
 If we isolate this random experiment, it is as if we take $Y_n = \max\{X_1,...,X_{n+1}\}+X_Y$, where $Y\sim U\left(\{1,2,...,n+1\}\mysetminus \{\mu\}\right)$, $\mu \in \argmax\{X_1,...,X_{n+1}\}$,  
$X_i\sim U[0,1] ~\forall i\in [n+1]$, and all the $X_i$s are independent. We estimate now the probability $\mathrm{P}(Y_n \le a)$ for $1<a<2$. We will set $a$ to a particular value in this interval later on. In fact, we bound this probability using the corresponding probability for $Z_n = \max\{X_1,...,X_{n+1}\}+X_{Y'}$, where $Y'\sim U\{1,2,...,n+1\}$. For $Z_n$ we have
\begin{IEEEeqnarray*}{rCl}
\mathrm{P}(Z_n \le a) & = & \sum_{i=1}^{n+1} \int_0^a \mathrm{P}\left( \max_{1\le j\le n+1}X_j \le t \wedge Y=i \wedge X_i \le a-t \right) dt\\
& = &  (n+1) \int_0^a \mathrm{P}\left(  \max_{1\le j\le n+1}X_j \le t \wedge Y=1 \wedge X_1 \le a-t\right)  dt \\
& = &  \int_0^a \mathrm{P}\left( \max_{1\le j\le n+1}X_j \le t \wedge X_1 \le a-t \right) dt  \\
& = &  \int_0^a \mathrm{P}(X_1 \le t \wedge X_1 \le a-t \wedge X_2 \le t\wedge \ldots \wedge X_{n+1} \le t)dt\\
& = &  \int_0^{a/2} \mathrm{P}(X_1 \le t \wedge X_2 \le t\wedge \ldots \wedge X_{n+1} \le t)dt + \\
&& +\:\int_{a/2}^1 \mathrm{P}(X_1 \le a-t \wedge X_2 \le t\wedge \ldots \wedge X_{n+1} \le t)dt + \\ 
&& +\:\int_1^{a} \mathrm{P}(X_1 \le a-t \wedge X_2 \le t\wedge \ldots \wedge X_{n+1} \le t)dt \\
& = &  \int_0^{a/2} t^{n+1} dt + \int_{a/2}^1 (a-t)t^n dt + \int _1^a (a-t)dt \,.
\end{IEEEeqnarray*}
Also, by the definition of $Y'$ we have $\mathrm{P}(Y'\notin\argmax\{X_1,...,X_{n+1}\}) = {n}/({n+1})$. 
Therefore, for $Y_n$ we get
\begin{IEEEeqnarray*}{rCl}
\mathrm{P}(Y_n \le a) & = & \mathrm{P}(Z_n \le a \ \mid Y'\notin\argmax\{X_1,...,X_{n+1}\}) \\
& = &  \frac{\mathrm{P}(Z_n \le a\ \wedge\ Y'\notin\argmax\{X_1,...,X_{n+1}\})}{\mathrm{P}(Y'\notin\argmax\{X_1,...,X_{n+1}\})} \\
& \le &  \frac{\mathrm{P}(Z_n \le a)}{\mathrm{P}(Y'\notin\argmax\{X_1,...,X_{n+1}\})} = \frac{n+1}{n} \mathrm{P}(Z_n \le a)  \\
& = & \frac{n+1}{n} \left( \int_0^{a/2} t^{n+1} dt + \int_{a/2}^1 (a-t)t^n dt + \int _1^a (a-t)dt \right)\,,
\end{IEEEeqnarray*}
where for the inequality we used the fact that $\mathrm{P}(A\cap B)\le \mathrm{P}(A)$ for any events $A, B$.

A similar analysis for the $j$th agent yields
\[\mathrm{P}(Y_j \le a) \le \frac{2n-j+1}{n} \left(\int_0^{a/2} t^{2n-j+1} dt + \int_{a/2}^1 (a-t)^{n-j+1}t^n dt +  \int _1^a (a-t)^{n-j+1} dt \right).\]
In the more general case where $m=2n+\kappa(n)$,  $0 \le\kappa(n)< n$, we have a similar calculation for the agents that receive only two 
items in the Greedy Round-Robin algorithm, as well as for the first two items of the first $\kappa(n)$ agents (who receive three items each). 
Let $Y_i$ be the total value agent $i$ receives, and $W_i$ be the value of her first two items. Of course, for the last $2n$ players, $Y_i = W_i$. Also, recall
 that $\sum_{j=1}^{m}v_{ij}=v_i(M)$. We now relate the probability that we are interested in estimating, with the probabilities 
 $\mathrm{P}(Y_i \le a)$ that we have already bounded. We will then proceed by setting $\alpha$ appropriately.
 We have
\begin{IEEEeqnarray*}{rCl}
\mathrm{P}\bigg( \exists i \text{ such that } Y_i &<& \frac{1}{n}\sum_{j=1}^{m}v_{ij}\bigg)  
			\le \sum_{i=1}^{n} \mathrm{P} \left( Y_i<\frac{v_i(M)}{n}\right)  \\ 
&=& \sum_{i=1}^{n} \mathrm{P} \left( Y_i<\min\bigg\lbrace \frac{v_i(M)}{n}, a\bigg\rbrace 
										\vee \frac{v_i(M)}{n}>\max\left\lbrace Y_i, a\right\rbrace \right)  \\
&\le& \sum_{i=1}^{n} \mathrm{P} \left( Y_i<\min\bigg\lbrace \frac{v_i(M)}{n}, a\bigg\rbrace\right) 
+\sum_{i=1}^{n} \mathrm{P} \left( \frac{v_i(M)}{n}>\max\left\lbrace Y_i, a\right\rbrace \right)  \\
& \le & \sum_{i=1}^{n} \mathrm{P}(Y_i<a) + \sum_{i=1}^{n} \mathrm{P}\left( \frac{v_i(M)}{n}>a \right) \,.
\end{IEEEeqnarray*}
To upper bound the first sum we use the $W_i$s, i.e.,  we do not take into account the third item that the first $\kappa(n)$ agents receive.  By the definition of $Y_i, W_i$, for these first $\kappa(n)$ agents we have $\mathrm{P}(Y_i  < a) \le \mathrm{P}( W_i<a)$, while for the remaining agents we have $\mathrm{P}(Y_i  < a) = \mathrm{P}( W_i<a)$. Note that the bounds for $\mathrm{P}(Y_i \le a)$ calculated above, here hold for $\kappa(n)+1\le i \le n$. For $1\le i \le \kappa(n)$ the same bounds hold for $\mathrm{P}(W_i \le a)$.
\begin{IEEEeqnarray*}{rCl}
\sum_{i=1}^{n}\!&\mathrm{P}&\!\!(Y_i  < a) \le 
            \sum_{i=1}^{\kappa(n)} \mathrm{P}( W_i<a) + 
                        \sum_{i=\kappa(n)+1}^{n} \mathrm{P}(Y_i<a) \\
&\le& \sum_{i=1}^{n} \textstyle\frac{m-i+1}{n} \displaystyle \bigg( \int_0^{a/2} t^{m-i+1} dt + \int_{a/2}^1  (a-t)^{n+\kappa(n)-i+1}t^n dt + 
                         \int _1^a  (a-t)^{n+\kappa(n)-i+1} dt \bigg) \\
&\le& 3 \sum_{j=1}^{n} \bigg( \int_0^{a/2} t^{n+\kappa(n)+j} dt + \int_{a/2}^1 (a-t)^{\kappa(n)+j}t^n dt + 
            \int _1^a (a-t)^{\kappa(n)+j} dt \bigg) \\
&=& 3 \left(\,\sum_{j=1}^{n}{\frac{(a/2)^{n+\kappa(n)+j+1}}{n+\kappa(n)+j+1}}  +
             \sum_{j=1}^{n}{\int_{a/2}^1 (a-t)^{\kappa(n)+j}t^n dt} +  \sum_{j=1}^{n}{\int _0^{a-1} u^{\kappa(n)+j}}  du \right) \,.
\end{IEEEeqnarray*}

We are going to bound each sum separately. We set $a=1+\frac{\kappa(n)}{2n}+\sqrt{\frac{3\ln n}{n}}= \frac{m}{2n}+\sqrt{\frac{3\ln n}{n}}$.
Note that for $n\ge 46$ we have $a\in (1, 2)$. Consider the first sum: 
\begin{IEEEeqnarray*}{rCl}
\sum_{j=1}^{n}{\frac{(a/2)^{n+\kappa(n)+j+1}}{n+\kappa(n)+j+1}} & \le & \frac{(a/2)^{n+\kappa(n)+2}}{n+\kappa(n)+2}\sum_{i=0}^{n-1}(a/2)^{i}  \\
&<&  \frac{1}{n+\kappa(n)+2} \cdot \frac{(a/2)^{n+\kappa(n)+2}}{1-a/2} =  O(1/n)\,,
\end{IEEEeqnarray*}
where we got $O(1/n)$ because the bound is at most 
$\frac{3}{n}$ for $n\ge 57$
and for any value of $\kappa(n)$. 

Next, we deal with the second sum:
\begin{IEEEeqnarray*}{rCl}
\sum_{j=1}^{n}{\int_{a/2}^1 (a-t)^{\kappa(n)+j}t^n dt} &<&  \int_{a/2}^1 t^n \left(  \sum_{j=0}^{\infty} (a-t)^{\kappa(n)+j} \right)  dt \le 
\int_{a/2}^1 t^n (a/2)^{\kappa(n)} \frac{1}{1-a+t}  dt  \\
&\le& \frac{(a/2)^{\kappa(n)}}{1-a/2} \int_{a/2}^1 t^n  dt = \frac{(a/2)^{\kappa(n)}}{1-a/2} \left( \frac{1-(a/2)^{n+1}}{n+1} \right) = O(1/n)\,.
\end{IEEEeqnarray*}
Here, for 
$n\ge 58$ the bound is at most $\frac{10}{n}$
for any $\kappa(n)$. 

Finally, for the third sum, we rewrite it as 
\[ \sum_{j=1}^{n}{\int _0^{a-1} u^{\kappa(n)+j}} du = \sum_{j=1}^{n} {\frac{(a-1)^{\kappa(n)+j+1}}{\kappa(n)+j+1}}
= \sum_{i=\kappa(n)+2}^{n+\kappa(n)+1} {\frac{(a-1)^{i}}{i}}\,. \]
We are going to bound each term separately. Consider the case where $\kappa(n)\ge 5\sqrt{n}$. For 
$n\ge 64$, it can be shown that
$\frac{1}{5}\left( \frac{\kappa(n)}{2n}+\sqrt{\frac{3\ln n}{n}}\,\right) ^{5\sqrt{n}}<\frac{10}{n^{3/2}}$.
So, 
\[\sum_{i=\kappa(n)+2}^{n+\kappa(n)+1} {\frac{(a-1)^{i}}{i}} \le \sum_{i=1}^{n}{\frac{(a-1)^{\kappa(n)}}{\kappa(n)}}
\le n\cdot{\frac{\left( \frac{\kappa(n)}{2n}+\sqrt{\frac{3\ln n}{n}}\,\right) ^{5\sqrt{n}}}{5\sqrt{n}}} \le n \cdot\frac{10}{n^{2}}= \frac{10}{n}\,. \]
On the other hand, when $\kappa(n)< 5\sqrt{n}$, we have $a-1<\frac{2.5+\sqrt{3\ln n}}{\sqrt{n}}$. For $n\ge 59$ and $j \ge 10$ it can be shown
that $\frac{1}{j}\left( \frac{2.5+\sqrt{3\ln n}}{\sqrt{n}}\right)^j< \frac{30}{n^{2}}$. Of course, for $3\le j\le 9$ it is true that 
$\frac{1}{j}\left( \frac{2.5+\sqrt{3\ln n}}{\sqrt{n}}\right)^j=o(1/n)$, and  particularly 
for 
$n\ge 59$ the sum $\sum_{i=3}^{9}\frac{1}{j}\Big( \frac{2.5+\sqrt{3\ln n}}{\sqrt{n}}\Big)^j$ is bounded by $\frac{25}{n}$. 
In general, it is to be expected to have relatively large hidden constants when $m$ is  very close to $2n$. This changes quickly though; when $\kappa(n)>21$ the whole sum is less than $1/n$. In any case, if $\kappa(n)>0$
\begin{IEEEeqnarray*}{rCl}
\sum_{i=\kappa(n)+2}^{n+\kappa(n)+1} {\frac{(a-1)^{i}}{i}} & \le & \sum_{i=3}^{n+2}{\frac{(a-1)^{i}}{i}} \le
 \sum_{i=3}^{9}{\textstyle\frac{1}{i}\Big( \frac{2.5+\sqrt{3\ln n}}{\sqrt{n}}\Big)}^i +
  \sum_{i=10}^{n+2}\textstyle{\frac{1}{i}\Big(\frac{2.5+\sqrt{3\ln n}}{\sqrt{n}}\Big)}^i \le\\
&\le&  O(1/n) + (n-7)\frac{30}{n^{2}}= O(1/n)\,.
\end{IEEEeqnarray*}
However, if $\kappa(n)=0$, we have
\[\sum_{i=2}^{n+1} {\frac{(a-1)^{i}}{i}} = {\textstyle \left( \frac{2.5+\sqrt{3\ln n}}{\sqrt{n}}\right)^2} + \sum_{i=3}^{n+1}{\frac{(a-1)^{i}}{i}}= O\left( \textstyle\frac{\log n}{n}\right)   \,. \]

So far, we have $\sum_{i=1}^{n} \mathrm{P}(Y_i<a)=O(1/n)$ (or $O(\log n /n)$ when $m=2n$). In order to complete the proof for this case we use Hoeffding's inequality 
to bound the probability that the average of the values for any agent is too large.
\begin{IEEEeqnarray*}{rCl}
\sum_{i=1}^{n} \mathrm{P}\left( \frac{v_i(M)}{n}>a\right)  &\le & n \cdot \mathrm{P}\left( \frac{v_1(M)}{n}>a\right)  
            = n \cdot \mathrm{P}\left( \frac{v_1(M)}{m}>\frac{n}{m}\left( \frac{m}{2n}+\textstyle\sqrt{\frac{3\ln n}{n}}\right) \right)  \\
& = & n \cdot \mathrm{P}\left( \frac{v_1(M)}{m}-\frac{1}{2} > \frac{n}{m}\textstyle\sqrt{\frac{3\ln n}{n}}\right)  \le n \cdot e^{-2m \left(  \frac{n}{m}\sqrt{\frac{3\ln n}{n}} \right) ^2 } \\
& = & n \cdot e^{-2 \frac{n}{m}\cdot 3\ln n}    \le n \cdot e^{-2 \ln n}    = \frac{1}{n}\,.
\end{IEEEeqnarray*}
Hence, 
\[\mathrm{P}\left( \exists i \text{ such that } Y_i<\frac{v_i(M)}{n}\right) = \begin{cases} O\Big(\frac{\log n}{n}\Big) & \text{ if } m=2n \\ O\Big(\frac{1}{n}\Big) & \text{ if } 2n<m<3n \end{cases} \,.\]

The remaining cases are for $m\ge 3n$.  
We give the proof for $m\ge 4n$. The cases for $3n\le m <3.5n$ and $3.5n\le m <4n$ differ in small details but they essentially follow the same analysis. We briefly discuss these cases at the end of the proof. 

Assume that $kn\le m < (k+1)n, k\ge 4$. 
We focus on the agent that choses last, i.e., agent 
$n$, who has the smallest expected value. She  gets exactly $k$ items, and like before let $Y_n$ be
the total value she receives. 
In order to bound $\mathrm{P}\left(Y_n <\beta\right)$ we introduce the random variables $Z_n$ and $W_n$. Consider the following random experiment involving the independent random variables $X_1, \ldots, X_{m-n+1}$, $X_i\sim U[0,1]\ \forall i\in[m-n+1]$. Given a realization of the $X_i$s, i.e., some values $x_1, \ldots, x_{m-n+1}$ in $[0,1]$, $Z_n$ is defined similarly to $Y_n$:
\begin{itemize}
\item Initially, $Z_n=0$.
\item While there are still $x_i$s left, take the maximum of the remaining $x_i$s, add it to $Z_n$, remove it from the available numbers, and then  remove the $x_i$s with the $n-1$ highest indices.
\item Return $Z_n$.
\end{itemize}
On the other hand, $W_n=\sum_{i=1}^{k-1}X_{(m+1-in, m-i(n-1))}$, where $X_{(j, t)}$ is the $j$th order statistic of $X_1, \ldots, X_{t}$. That is, $W_n$ is defined as the sum of the largest of all $x_i$s, the second largest of the first $m-n+1$ $x_i$s, the third largest of the first $m-2n+2$ $x_i$s, and so on. 

It is not hard to see that always $W_n\le Z_n$ (in fact, each term of $W_n$ is less than or equal to the corresponding term of $Z_n$) and that $Z_n$ follows the same distribution as $Y_n$. So, 
$\mathrm{P}\left(Y_n<\beta\right) = \mathrm{P}\left(Z_n <\beta\right) \le  \mathrm{P}\left(W_n <\beta\right)$.
Using the fact that the $i$th order statistic in a sample of size $\ell$ drawn independently from $U[0, 1]$ has expected value $\frac{i}{\ell+1}$ and variance $\frac{i(\ell-i+1)}{(\ell +1)^2 (\ell +2)}$ \citep{Gentle}, we get
\begin{IEEEeqnarray*}{rCl}
\mathrm{E}[W_n]&=& \frac{m-n+1}{m-n+2}+\frac{m-2n+1}{m-2n+3}+\ldots+ \frac{m-(k-1)n+1}{m-(k-1)n+k}  \\
&\ge& \frac{(k-1)n+1}{(k-1)n+2}+\frac{(k-2)n+1}{(k-2)n+3}+\ldots+ \frac{n+1}{n+k} \\
&>& k-1 -\frac{1}{(k-1)n}-\frac{2}{(k-2)n}-\ldots- \frac{k-1}{n} > k-1 -\frac{(k-1)H_{k-1}}{n}\,.
\end{IEEEeqnarray*}
Moreover, if $X_i' = X_{(m+1-in, m-i(n-1))}$ we have
\begin{IEEEeqnarray*}{rCl}
\sigma_{W_n}^2 = \mathrm{Var}(W_n) &=& \sum_{i=1}^{k-1}\sum_{j=1}^{k-1} \mathrm{Cov}(X_i, X_j) \le 
                        \sum_{i=1}^{k-1}\sum_{j=1}^{k-1} \sqrt{\mathrm{Var}(X_i)\mathrm{Var}(X_j)} \le \left(\sum_{i=1}^{k-1} \sqrt{\mathrm{Var}(X_i)}\right)^2 \\
&<& \left(\sum_{i=1}^{k-1} \frac{\sqrt{i\,}}{m-i n+ i +1}\right)^2 <  \left(\sqrt{k-1}\sum_{i=1}^{k-1} \frac{1}{(k-i)n}\right)^2 = \frac{(k-1)H_{k-1}^2}{n^2} \,,
\end{IEEEeqnarray*}
where $H_{k-1}$ is the $(k-1)$-th harmonic number. Now we can bound the probability that any agent receives value
less than $1/n$ of her total value.
\[ \mathrm{P}\left( Y_i<\frac{v_i(M)}{n}\right)  \le  \mathrm{P}\left( Y_n<\frac{v_n(M)}{n}\right)  \le \mathrm{P}\left(Y_n<\frac{13k}{20}\right) + \mathrm{P}\left( \frac{v_n(M)}{n} > \frac{13k}{20}\right) \,.\]
Next, using Chebyshev's inequality we have 
\begin{IEEEeqnarray*}{rCl}
\mathrm{P}\left(Y_n<\frac{13k}{20}\right)& \le & \mathrm{P}\left(W_n<\frac{13k}{20}\right) = 
                                        \mathrm{P}\left(\mathrm{E}[W_n] -W_n > \mathrm{E}[W_n]-\frac{13k}{20}\right) \\
&\le& \mathrm{P}\left(\left|\mathrm{E}[W_n] -W_n \right| > k-1 -\frac{(k-1)H_{k-1}}{n}-\frac{13k}{20}\right) \\
&\le& \mathrm{P}\left(\left|\mathrm{E}[W_n] -W_n \right| > \frac{\frac{7k}{20}-1 -\frac{(k-1)H_{k-1}}{n}}{\frac{\sqrt{k-1} H_{k-1}}{n}}\sigma_{W_n}\right)\\
&\le& \frac{(k-1) H_{k-1}^2}{\left(\left(\frac{7k-20}{20}\right)n- (k-1) H_{k-1}\right)^2}\,.
\end{IEEEeqnarray*}
On the other hand, using Hoeffding's inequality,
\begin{IEEEeqnarray*}{rCl}
\mathrm{P}\left( \frac{v_n(M)}{n} > \frac{13k}{20}\right) & = & 
        \mathrm{P}\left( \frac{v_n(M)}{m} -\frac{1}{2} > \frac{n}{m}\frac{13k}{20}-\frac{1}{2}\right)  \\
& \le & \mathrm{P}\left( \frac{v_n(M)}{m} -\frac{1}{2} > \frac{13k}{20(k+1)}-\frac{1}{2}\right)  \\
&\le& e^{-2m\left( \frac{3k-10}{20(k+1)} \right)^2} \le e^{-2kn\left( \frac{3k-10}{20(k+1)} \right)^2}\,.
\end{IEEEeqnarray*}
Finally, we take a union bound to get
\begin{IEEEeqnarray*}{rCl}
\mathrm{P}\left( \exists i \text{ s.t. } Y_i<\frac{v_i(M)}{n}\right) &\le & \sum_{i=1}^{n} \mathrm{P}\left( Y_i<\frac{v_i(M)}{n}\right)\\ 
&\le& n\left( {\textstyle \frac{(k-1) H_{k-1}^2}{\left(\left(\frac{7k-20}{20}\right)n- (k-1) H_{k-1}\right)^2}} + e^{-2kn\left( \frac{3k-10}{20(k+1)} \right)^2} \right)= O(1/n) \,.
\end{IEEEeqnarray*}
%
%
The exact same proof works when $3n\le m <3.5n$, but instead of $\frac{3k-10}{20(k+1)}$ in Hoeffding's inequality, we have $\frac{3\cdot 3-5}{20(3+0.5)}$ and of course we should adjust $\mathrm{E}[W_n]$ and $\mathrm{Var}(W_n)$ accordingly. When $3.5n\le m <4n$ on the other hand, we need to consider three items in $W_n$ instead of two, since two items are not  enough anymore to guarantee separation of $Y_i$ and $\frac{1}{n}\sum_{j=1}^{m}v_{ij}$ with high probability. That said, the proof is the same, but we should adjust $\mathrm{E}[W_n]$ and $\mathrm{Var}(W_n)$, and instead of $\frac{13k}{20}=\frac{39}{20}$ we may choose $2.5$. 
\end{proof}

We now state a similar result for any $m$, generalizing Theorem \ref{thm:random-prop} that only holds when $m\ge 2n$. We use a  modification
of Greedy Round-Robin. While $m<2n$, the algorithm picks any agent uniformly at random and gives her only her 
``best'' item (phase 1). 
When the number of available items becomes two times the number of active agents, the algorithm proceeds as 
usual (phase 2). We note that while for $m\ge 2n$ Theorem \ref{thm:random-prop} gives the stronger guarantee of $\frac{v_i(M)}{n}$ for each agent $i$, here we can only have a guarantee of $\mms_i(n, M)$.

\begin{theorem}\label{thm:random-mms}
	Let $N=[n]$, $M=[m]$, and the $v_{ij}$s be as in Theorem \ref{thm:random-prop}. Then, for any $m$ and large enough $n$, the Modified Greedy Round-Robin algorithm allocates to each agent $i$ 
	a set of items of total value at least $\mms_i(n, M)$ with probability $1-o(1)$. The $o(1)$ term is $O(1/n)$ when $m>2n$ and $O(\log n / n)$ when $m \le 2n$.
\end{theorem}

\begin{proof}
If $m\ge 2n$ then this is a corollary of Theorem \ref{thm:random-prop}. When $m<2n$, then for any agent $i$ we have 
$\max_j\{v_{ij}\}\ge \mms_i(n, M)$. So the first agent that receives only her most 
valuable item has total value at least $\mms_i(n, M)$. If $N_a, M_a$ are the sets of remaining agents and items respectively, 
after several agents were assigned one item in phase 1 of the algorithm, then by Lemma \ref{lem:monotonicity}, for any agent 
$i\in N_a$, we have $\mms_i(|N_a|, M_a) \ge \mms_i(n, M)$. If $|M_a|<2|N_a|$ it is also true that 
$\max_{j\in M_a}v_{ij}\ge \mms_i(|N_a|, M_a)$, so correctness of phase 1 follows by induction. If 
$|M_a|=2|N_a|$, then by Theorem \ref{thm:random-prop} phase 2 guarantees that with high probability each agent $i\in N_a$
will receive a set of items with total value at least $\frac{1}{|N_a|}v_{i}(M_a) \ge \mms_i(|N_a|, M_a) \ge \mms_i(n, M)$.
\end{proof}

\noindent{\bf Remark 1.}
The implicit constants in the probability bounds of Theorems \ref{thm:random-prop} and \ref{thm:random-mms} depend heavily on $n$ and $m$, as well as on the point one uses to separate $Y_i$ and $\frac{1}{n}\sum_{j=1}^{m}v_{ij}$ in the proof of Theorem \ref{thm:random-prop}.
Our analysis gives good bounds for the case $2n\le m<3n$ without requiring very large values for $n$ (especially when $\kappa(n)$ in the proof of Theorem \ref{thm:random-prop} is not small). 
For example, if $m=2.4n$ an appropriate adjustment of our bounds gives a $o(1)$ term less than $1.7/n$ for $n\ge 41$.
When we switch from the detailed analysis of the $2n\le m<3n$ case to the sloppier general treatment for $m \ge 3n$, there is definitely 
some loss,
e.g., for $m=4n$ we get that the $o(1)$ term is less than $130/n$ for $n> 450$. This is corrected relatively quickly as $m$ grows, e.g., for $m=13n$ the $o(1)$ term can be made less than $8/n$ for $n\ge 59$. One can significantly improve the constants by breaking the interval $kn\le m<(k+1)n$ into smaller intervals (not unlike the $3n\le m<3.5n$ case).\medskip

Theorems \ref{thm:random-prop} and \ref{thm:random-mms} may leave the impression that $n$ has to be large. 
Actually, 
there is no reason why we cannot consider $n$ fixed and let $m$ grow. Following closely the proof of Theorem \ref{thm:random-prop}
for $m \ge 4n$, we get the next corollary. Notice that now we can use $\mathrm{E}[W_n] \ge 0.7k$ and $\sigma_{W_n}^2 < k$.
\begin{corollary}\label{thm:random-prop-m}
	Let $N=[n]$, $M=[m]$, and the $v_{ij}$s be as in Theorem \ref{thm:random-prop}. 
	Then, for fixed $n$ and large enough $m$, the Greedy Round-Robin algorithm allocates to each agent $i$ 
	a set of goods of total value at least $\frac{1}{n}\sum_{j=1}^m v_{ij}$ with probability $1-O(1/m)$. 
\end{corollary}

\section{Conclusions}
The most interesting open question is undoubtedly whether one can improve on the $2/3$-appro\-xi\-ma\-tion. 
Going beyond $2/3$ seems to require a drastically different approach.  One idea that may deserve further exploration is to pick in each step of Algorithm \ref{fig:alg-pw_rec}, the best out of all possible matchings (and not just an arbitrary matching as is done in line~\ref{line:allocate} of the algorithm). This is essentially what we exploit for the special case of $n=3$ agents. However, for a larger number of agents, this seems to result in a heavy case analysis without any visible benefits. In terms of non-combinatorial techniques, we are not currently aware of any promising LP-based approach to the problem. 

Even establishing better ratios for special cases could still provide new insights into the problem. It would be interesting, for example, to see if we can have an improved ratio for the special case studied in \citet{BS06} for the Santa Claus problem. In this case of additive functions, the value of a good $j$ takes only two distinct values, $0$ or $v_j$.
On the other hand, obtaining negative results seems to be an even more challenging task, given our probabilistic analysis and the results of related works. 
The negative results~\citep{KPW16,PW14} require  very elaborate constructions, which still do not yield an inapproximability factor far away from $1$.
Apart from improving the approximation quality, exploring practical aspects of our algorithms is another direction, see e.g.,~\citet{spliddit}. Finally, we have not addressed here the issues of truthfulness and mechanism design, a stimulating topic for future work, studied recently by~\citep{ABM16,ABCM17}. 
These works still leave several open questions regarding the approximability that can be achieved under truthfulness (without payments) for more than two agents. It is also not clear if more positive results can arise 
when payments are allowed. 
Similar mechanism design questions also remain open for a related problem studied by~\citet{MP11}.  


\section*{Acknowledgements}
Research co-financed by the European Union (European Social Fund - ESF) and Greek national funds through the Operational Program ``Education and Lifelong Learning'' of the National Strategic Reference Framework (NSRF) - Research Funding Program: ``THALES - Investing in knowledge society through the European Social Fund''.

\bibliographystyle{plainnat}
\bibliography{fairdivrefs}

\begin{thebibliography}{30}
\providecommand{\natexlab}[1]{#1}
\providecommand{\url}[1]{\texttt{#1}}
\expandafter\ifx\csname urlstyle\endcsname\relax
  \providecommand{\doi}[1]{doi: #1}\else
  \providecommand{\doi}{doi: \begingroup \urlstyle{rm}\Url}\fi

\bibitem[Amanatidis et~al.(2015)Amanatidis, Markakis, Nikzad, and
  Saberi]{AMNS15}
G.~Amanatidis, E.~Markakis, A.~Nikzad, and A.~Saberi.
\newblock Approximation algorithms for computing maximin share allocations.
\newblock In \emph{Automata, Languages, and Programming - 42nd International
  Colloquium, {ICALP} 2015, Proceedings, Part {I}}, pages 39--51, 2015.

\bibitem[Amanatidis et~al.(2016)Amanatidis, Birmpas, and Markakis]{ABM16}
G.~Amanatidis, G.~Birmpas, and E.~Markakis.
\newblock On truthful mechanisms for maximin share allocations.
\newblock In \emph{Proceedings of the 25th International Joint Conference on
  Artificial Intelligence, {IJCAI} 2016}, pages 31--37, 2016.

\bibitem[Amanatidis et~al.(2017)Amanatidis, Birmpas, Christodoulou, and
  Markakis]{ABCM17}
G.~Amanatidis, G.~Birmpas, G.~Christodoulou, and E.~Markakis.
\newblock Truthful allocation mechanisms without payments: Characterization and
  implications on fairness.
\newblock In \emph{{ACM} Conference on Economics and Computation, {EC} '17},
  pages 545--562, 2017.

\bibitem[Asadpour and Saberi(2007)]{AS07}
A.~Asadpour and A.~Saberi.
\newblock An approximation algorithm for max-min fair allocation of indivisible
  goods.
\newblock In \emph{ACM Symposium on Theory of Computing (STOC)}, pages
  114--121, 2007.

\bibitem[Aziz and MacKenzie(2016)]{AM16}
H.~Aziz and S.~MacKenzie.
\newblock A discrete and bounded envy-free cake cutting protocol for four
  agents.
\newblock In \emph{48th ACM Symposium on the Theory of Computing, {STOC 2016}},
  pages 454--464, 2016.

\bibitem[Bansal and Sviridenko(2006)]{BS06}
N.~Bansal and M.~Sviridenko.
\newblock The santa claus problem.
\newblock In \emph{ACM Symposium on Theory of Computing (STOC)}, pages 31--40,
  2006.

\bibitem[Barman and Murthy(2017)]{BM17}
S.~Barman and S.~K.~K. Murthy.
\newblock Approximation algorithms for maximin fair division.
\newblock In \emph{{ACM} Conference on Economics and Computation, {EC} '17},
  pages 647--664, 2017.

\bibitem[Bezakova and Dani(2005)]{BD05}
I.~Bezakova and V.~Dani.
\newblock Allocating indivisible goods.
\newblock \emph{ACM SIGecom Exchanges}, 5:\penalty0 11--18, 2005.

\bibitem[Bouveret and Lang(2011)]{BL11}
S.~Bouveret and J.~Lang.
\newblock A general elicitation-free protocol for allocating indivisible goods.
\newblock In \emph{Proceedings of the 22nd International Joint Conference on
  Artificial Intelligence, {IJCAI} 2011}, pages 73--78, 2011.

\bibitem[Bouveret and Lema{\^{\i}}tre(2016)]{BL16}
S.~Bouveret and M.~Lema{\^{\i}}tre.
\newblock Characterizing conflicts in fair division of indivisible goods using
  a scale of criteria.
\newblock \emph{Autonomous Agents and Multi-Agent Systems}, 30\penalty0
  (2):\penalty0 259--290, 2016.
\newblock A preliminary version of this work has appeared in {AAMAS} '14.

\bibitem[Bouveret et~al.(2016)Bouveret, Chevaleyre, and Maudet]{BCM16-survey}
S.~Bouveret, Y.~Chevaleyre, and N.~Maudet.
\newblock Fair allocation of indivisible goods.
\newblock In F.~Brandt, V.~Conitzer, U.~Endriss, J.~Lang, and A.~D. Procaccia,
  editors, \emph{Handbook of Computational Social Choice}, chapter~12.
  Cambridge University Press, 2016.

\bibitem[Brams and King(2005)]{BK05}
S.~J. Brams and D.~King.
\newblock Efficient fair division - help the worst off or avoid envy.
\newblock \emph{Rationality and Society}, 17\penalty0 (4):\penalty0 387--421,
  2005.

\bibitem[Brams and Taylor(1996)]{BT96}
S.~J. Brams and A.~D. Taylor.
\newblock \emph{Fair Division: from Cake Cutting to Dispute Resolution}.
\newblock Cambridge University press, 1996.

\bibitem[Budish(2011)]{Budish11}
E.~Budish.
\newblock The combinatorial assignment problem: Approximate competitive
  equilibrium from equal incomes.
\newblock \emph{Journal of Political Economy}, 119\penalty0 (6):\penalty0
  1061--1103, 2011.

\bibitem[Caragiannis et~al.(2016)Caragiannis, Kurokawa, Moulin, Procaccia,
  Shah, and Wang]{CKMPSW16}
I.~Caragiannis, D.~Kurokawa, H.~Moulin, A.~D. Procaccia, N.~Shah, and J.~Wang.
\newblock The unreasonable fairness of maximum nash welfare.
\newblock In \emph{{ACM} Conference on Economics and Computation, {EC} '16},
  pages 305--322, 2016.

\bibitem[Edmonds and Pruhs(2006)]{EP06-focs}
J.~Edmonds and K.~Pruhs.
\newblock Balanced allocations of cake.
\newblock In \emph{Symposium on Foundations of Computer Science (FOCS)}, pages
  623--634, 2006.

\bibitem[Even and Paz(1984)]{EP84}
S.~Even and A.~Paz.
\newblock A note on cake cutting.
\newblock \emph{Discrete Applied Mathematics}, 7:\penalty0 285--296, 1984.

\bibitem[Gentle(2009)]{Gentle}
J.E. Gentle.
\newblock \emph{Computational statistics}, volume 308.
\newblock Springer, 2009.

\bibitem[Hall(1935)]{Hall35}
P.~Hall.
\newblock On representatives of subsets.
\newblock \emph{Journal of the London Mathematical Society}, s1-10\penalty0
  (1):\penalty0 26--30, 1935.

\bibitem[Hoeffding(1963)]{Hoeffding63}
W.~Hoeffding.
\newblock Probability inequalities for sums of bounded random variables.
\newblock \emph{Journal of the American Statistical Association}, 58\penalty0
  (301):\penalty0 13--30, 1963.

\bibitem[Kurokawa et~al.(2016)Kurokawa, Procaccia, and Wang]{KPW16}
D.~Kurokawa, A.~D. Procaccia, and J.~Wang.
\newblock When can the maximin share guarantee be guaranteed?
\newblock In \emph{30th AAAI Conference on Artificial Intelligence, {AAAI}
  2016}, 2016.

\bibitem[Markakis and Psomas(2011)]{MP11}
E.~Markakis and C.{-}A. Psomas.
\newblock On worst-case allocations in the presence of indivisible goods.
\newblock In \emph{7th Workshop on Internet and Network Economics (WINE 2011)},
  pages 278--289, 2011.

\bibitem[Moulin(1990)]{Moulin90}
H.~Moulin.
\newblock Uniform externalities: Two axioms for fair allocation.
\newblock \emph{Journal of Public Economics}, 43\penalty0 (3):\penalty0
  305--326, 1990.

\bibitem[Procaccia(2015)]{Procaccia14-survey}
A.~D. Procaccia.
\newblock Cake cutting algorithms.
\newblock In F.~Brandt, V.~Conitzer, U.~Endriss, J.~Lang, and A.D. Procaccia,
  editors, \emph{Handbook of Computational Social Choice}, chapter~13.
  Cambridge University Press, 2015.

\bibitem[Procaccia and Wang(2014)]{PW14}
A.~D. Procaccia and J.~Wang.
\newblock Fair enough: guaranteeing approximate maximin shares.
\newblock In \emph{{ACM} Conference on Economics and Computation, {EC} '14},
  pages 675--692, 2014.

\bibitem[Robertson and Webb(1998)]{RW98}
J.~M. Robertson and W.~A. Webb.
\newblock \emph{Cake Cutting Algorithms: be fair if you can}.
\newblock AK Peters, 1998.

\bibitem[Spliddit(2015)]{spliddit}
Spliddit.
\newblock Provably fair solutions.
\newblock \url{http://www.spliddit.org/}, 2015.

\bibitem[Steinhaus(1948)]{Steinhaus48}
H.~Steinhaus.
\newblock The problem of fair division.
\newblock \emph{Econometrica}, 16:\penalty0 101--104, 1948.

\bibitem[Woeginger(1997)]{Woeginger97}
G.~Woeginger.
\newblock A polynomial time approximation scheme for maximizing the minimum
  machine completion time.
\newblock \emph{Operations Research Letters}, 20:\penalty0 149--154, 1997.

\bibitem[Woeginger and Sgall(2007)]{WS07}
G.~Woeginger and J.~Sgall.
\newblock On the complexity of cake cutting.
\newblock \emph{Discrete Optimization}, 4\penalty0 (2):\penalty0 213--220,
  2007.

\end{thebibliography}


\end{document}